\documentclass[12pt]{amsart}
\usepackage{txfonts}      
\usepackage{amssymb}
\usepackage{eucal}
\usepackage{amsmath}
\usepackage{amscd}
\usepackage[dvips]{color}
\usepackage{multicol}
\usepackage[all]{xy}           
\usepackage{graphicx}
\usepackage{color}
\usepackage{colordvi}
\usepackage{xspace}
\usepackage{tikz}
\usepackage{makecell}
\usepackage{slashed}

\usepackage{ifpdf}
\ifpdf
\usepackage[colorlinks,final,backref=page,hyperindex]{hyperref}
\else
\usepackage[colorlinks,final,backref=page,hyperindex,hypertex]{hyperref}
\fi

\usepackage[active]{srcltx} 

\usepackage{bbding}
\usepackage[shortlabels]{enumitem}
\usepackage{tikz-cd}

\makeatletter

\newtheorem{thm}{Theorem}[section]
\newtheorem{lem}[thm]{Lemma}
\newtheorem{cor}[thm]{Corollary}
\newtheorem{pro}[thm]{Proposition}
\theoremstyle{definition}   
\newtheorem{defi}[thm]{Definition}

\newtheorem{ex}[thm]{Example}
\newtheorem{rmk}[thm]{Remark}

\newcommand{\nc}{\newcommand}
\newcommand{\delete}[1]{}
\nc{\eval}[1]{\Big|_{#1}}

\nc{\mlabel}[1]{\label{#1}}  
\nc{\mcite}[1]{\cite{#1}}  
\nc{\mref}[1]{\ref{#1}}  
\nc{\meqref}[1]{\eqref{#1}}  
\nc{\mbibitem}[1]{\bibitem{#1}} 

\delete{
\nc{\mlabel}[1]{\label{#1}{\hfill \hspace{1cm}{\tt{{\ }\hfill(#1)}}}}
\nc{\mcite}[1]{\cite{#1}{{\tt{{\ }(#1)}}}}  
\nc{\mref}[1]{\ref{#1}{{\tt{{\ }(#1)}}}}  
\nc{\meqref}[1]{\eqref{#1}{{\tt{{\ }(#1)}}}}  
\nc{\mbibitem}[1]{\bibitem[\tt #1]{#1}} 
}

\newcommand {\emptycomment}[1]{}

\nc{\vep}{\varepsilon}

\newcommand{\one}{{\mathbf{1}}}

\nc{\oprn}{\theta}
\nc{\Oprn}{\Theta}

\nc{\tforall}{\ \ \text{for all }}

\nc{\calo}{\mathcal{O}}
\nc{\oop}{$\mathcal{O}$-operator\xspace}
\nc{\oops}{$\mathcal{O}$-operators\xspace}
\nc{\mrho}{{\bm{\varrho}}}
\nc{\emk}{\mathbf{K}}
\nc{\invlim}{\displaystyle{\lim_{\longleftarrow}}\,}
\nc{\ot}{\otimes}

\newcommand{\B }{\mathfrak{B}}
\newcommand{\lon }{\,\rightarrow\,}
\newcommand{\be }{\begin{equation}}
\newcommand{\ee }{\end{equation}}

\newcommand{\g}{\mathfrak g}

\newcommand{\G}{{\mathbb G}}
\newcommand{\A}{{\mathbb A}}
\newcommand{\BB}{{\mathbb B}}
\newcommand{\RR}{{\mathbb R}}

\newcommand{\KK}{{\mathbb K}}
\newcommand{\HH}{\mathbb H}

\newcommand{\huaB}{\mathcal{B}}

\newcommand{\huaF}{\mathcal{F}}

\newcommand{\huaU}{\mathcal{U}}

\newcommand{\huaP}{\mathcal{P}}
\nc{\calp}{\mathcal{P}}

\newcommand{\CWM}{C^{\infty}(M)}

\newcommand{\frka}{\mathfrak a}

\newcommand{\frkd}{\mathfrak d}

\newcommand{\frkX}{\mathfrak X}

\newcommand{\Courant}[1]{\left\llbracket  #1\right\rrbracket }

\newcommand{\dev}{\mathfrak{D}}

\newcommand{\Id}{{\rm{Id}}}

\newcommand{\br}[1]{   [ \cdot,    \cdot  ]   }

\newcommand{\dM}{\mathrm{d}}

\newcommand{\Der}{\mathrm{Der}}
\newcommand{\Lie}{\mathrm{Lie}}
\newcommand{\Diff}{\mathrm{Diff}}

\newcommand{\PG}{\mathsf{PGpd}}

\newcommand{\FF}{\mathsf{F}}
\newcommand{\GG}{\mathsf{G}}

\newcommand{\BS}{\mathsf{NBQ}}
\nc{\NBS}{\mathsf{NBS}}
\newcommand{\BG}{\mathsf{BGpd}}

\newcommand{\SLB}{\mathsf{SLB}}

\newcommand{\Aut}{\mathrm{Aut}}

\newcommand{\gl}{\mathfrak {gl}}

\newcommand{\Quiv}{\mathrm {Quiv}}

\newcommand{\inv}{{\mathrm{inv}}}

\newcommand{\T}{\mathbb{T}}

\nc{\CV}{\mathbf{C}}

\begin{document}

\title[Post-groupoids and the Yang-Baxter equation]{Post-groupoids and   quiver-theoretical solutions
of the Yang-Baxter equation}

\author{Yunhe Sheng}
\address{Department of Mathematics, Jilin University, Changchun 130012, Jilin, China}
\email{shengyh@jlu.edu.cn}

\author{Rong Tang}
\address{Department of Mathematics, Jilin University, Changchun 130012, Jilin, China}
\email{tangrong@jlu.edu.cn}

\author{Chenchang Zhu}
\address{Mathematics Institute, Georg-August-University Gottingen, Bunsenstrasse 3-5 37073, Gottingen, Germany}
\email{czhu@gwdg.de}


\begin{abstract}
The notion of post-groups was introduced by Bai, Guo and the first two authors recently, which are the global objects corresponding to post-Lie algebras, equivalent to skew-left braces, and can be used to construct set-theoretical solutions of the Yang-Baxter equation.
In this paper, first we introduce the notion of post-groupoids, which consists of a group bundle and some other structures satisfying some compatibility conditions. Post-groupoids reduce to post-groups if the underlying base is one point. An action of a group on a set gives rise to the natural example of post-groupoids. We show that a post-groupoid gives rise to a groupoid (called the Grossman-Larson groupoid), and an action on the original group bundle. Then we introduce the notion of relative Rota-Baxter operators on a groupoid with respect to an action on a group bundle. A relative Rota-Baxter operator naturally gives rise to a post-groupoid and a matched pair of groupoids. One important application of post-groupoids is that they give rise to quiver-theoretical solutions of the Yang-Baxter equation on the quiver underlying the Grossman-Larson groupoid. We also introduce the notion of a skew-left bracoid, which consists of a group bundle and a groupoid satisfying some compatibility conditions. A skew-left bracoid reduces to a skew-left brace if the underlying base is one point. We give the one-to-one correspondence between post-groupoids and skew-left bracoids. Finally, we show that post-Lie groupoids give rise to post-Lie algebroids via differentiation.

\end{abstract}

\renewcommand{\thefootnote}{}
\footnotetext{2020 Mathematics Subject Classification.
22E60, 
16T25, 
17B38, 
65L99, 
}

\keywords{post-groupoid, Yang-Baxter equation, brace, post-Lie algebra, Rota-Baxter operator}

\maketitle

\tableofcontents

\allowdisplaybreaks

\section{Introduction}\mlabel{sec:intr}

\subsection{Post-Lie algebras  (algebroids) and their applications in numerical integration and regularity structures}

The notion of a post-Lie algebra was introduced by Vallette from his study of Koszul duality of operads in \mcite{Val}. Then Bai, Guo and Ni found that post-Lie algebras are the underlying structures of Rota-Baxter operators of weight 1 on Lie algebras, and can be applied to the study of generalized Lax pairs \mcite{BGN}.
Munthe-Kaas and Lundervold found that post-Lie algebras also naturally appear in differential geometry and play important roles in the study of numerical integration on  manifolds \mcite{ML}. Recently, Al-Kaabi, Ebrahimi-Fard,  Manchon and  Munthe-Kaas have used the  post-Lie framework to understand
 the framed Lie algebras and they partially  solved a  conjecture of Gavrilov  \cite{AFMM}. See \mcite{BG,CEO, Dotsenko,EMM, MQS} for more studies on  classifications of post-Lie algebras on certain Lie algebras, post-Lie Magnus expansion, Poincar\'e-Birkhoff-Witt theorems, factorization theorems of post-Lie algebras and post-symmetric brace algebras.

The combinatorial Hopf algebras on decorated trees have played an important role in the Taylor-type expansions of solutions of singular stochastic partial differential equations \mcite{BHZ,Hairer}. In fact, those Hopf algebras are sophisticated and difficult to handle. To better understand them,  Bruned and Katsetsiadis have unified  the construction of the combinatorial Hopf algebras in the theory of regularity structures as  the universal enveloping algebras of  post-Lie algebras on decorated trees and multi-indices \cite{BK}. Moreover, Jacques and Zambotti studied the post-Lie algebras of derivations and gave explicit computations on the structure group of regularity structures \cite{JZ}.

As the geometric generalization of a post-Lie algebra, the notion of a post-Lie algebroid was introduced by Munthe-Kaas and Lundervold in their study of numerical integration \cite{ML}.  Fl$\slashed{\text{o}}$ystad,  Manchon and Munthe-Kaas found that pre-Lie Rinehart algebras (the algebraic version of pre-Lie algebroids) are the algebraic foundations of aromatic B-series \cite{FMM}. Bronasco and Laurent studied the universal enveloping algebroid of the pre-Lie Rinehart algebra of aromatic trees and the  Hopf algebra structures  of ergodic
stochastic differential equations \cite{BL}. 
Since post-Lie algebras are the algebraic structure behind the regularity structures on Euclidean spaces, we are convinced that post-Lie algebroids will appear in the regularity structures on manifolds and vector bundles.
\vspace{-2mm}
\subsection{Set-theoretical and quiver-theoretical solutions of the Yang-Baxter equation, and related structures}
The Yang-Baxter equation grew out of Yang's study of exact
solutions of many-body problems \mcite{Ya} and Baxter's work of
the eight-vertex model \mcite{Bax}.  The equation has important applications in
many fields in mathematics and physics.
 It is very difficult to solve the Yang-Baxter equation, so Drinfeld suggested to study set-theoretical solutions of the Yang-Baxter equation \mcite{Dr}. Then Etingof-Schedler-Soloviev \mcite{ESS} and  Lu-Yan-Zhu \mcite{LYZ} studied the structure groups of such solutions.
It is natural to extend Drinfeld's question and ask for solutions to the   Yang-Baxter equation
in an arbitrary monoidal category. Andruskiewitsch addressed this problem in
the special case of the category of quivers and introduced the concept of  quiver-theoretical solutions of the
Yang-Baxter equation in \mcite{Andruskiewitsch-1}. Moreover, he used   braided groupoids to classify non-degenerate quiver-theoretical solutions of the Yang-Baxter equation, which generalized the result of Soloviev \cite{Soloviev}. Since post-groupoids are equivalent to braided groupoids, it is natural to use post-groupoids and their representations to classify  non-degenerate braided quivers.

As a generalization of radical rings, braces were introduced by Rump~\mcite{Ru} to construct set-theoretical solutions of the Yang-Baxter equation. Further studies were carried out in~\cite{CJO,G,Sm18}.
Recently, braces were generalized to skew-left braces by Guarnieri and Vendramin in \mcite{GV} to construct
non-degenerate and not necessarily involutive  solutions of the Yang-Baxter equation. Moreover, the radical and weight of skew-left braces were applied to study the structure groups of solutions of the Yang-Baxter equation \mcite{JKVV}. The nilpotency   of skew-left braces were investigated in \mcite{CSV}. The  Lazard correspondence for pre-Lie rings and braces has studied by Smoktunowicz \cite{Sm22b}. Under some natural completeness condition, Trappeniers extended the Lazard correspondence to post-Lie rings and skew-left braces \cite{Trappeniers}.

As the global objects corresponding to post-Lie algebras, the notion of post-groups was introduced in \cite{PostG}. The differentiation of a post-Lie group is a post-Lie algebra.  The  Butcher group and the $\huaP$-group of an operad $\huaP$ have  natural post-group structures. Moreover, it was shown that post-groups are equivalent to skew-left braces, and naturally give rise to set-theoretical solutions of the Yang-Baxter equation. In \cite{AFM}, the authors construct free post-groups from left-regular magmas, and observe that one can almost obtain a post group structure on the space of group-valued functions from an action of a  group on a set except the bijective condition, which motivate them to introduce the notion of  a weak post-group by dropping the bijective condition. See \cite{Fard1,Fard2} for more   applications of post-groups in free probability and SISO Affine Feedback Control Systems.

\vspace{-3mm}
\subsection{Rota-Baxter operators on Lie algebras and groups}
The notion of Rota-Baxter operators on associative algebras was introduced by Baxter in \cite{Bax}.  It has been found important applications in the
Connes-Kreimer's algebraic approach to renormalization of quantum field theory~\cite{CK}. See the book \cite{Gub} for more details.
In the Lie algebra context, Rota-Baxter operators (of weight 1) are in one-to-one correspondence  with
solutions of the modified Yang-Baxter equation \cite{STS} and give rise to factorizations of Lie algebras. Moreover, a Rota-Baxter operator naturally gives rise to a post-Lie
algebra, and play important roles in mathematical
physics~\cite{BGN}.

The concepts of Rota-Baxter operators on groups were introduced in \mcite{GLS} in their studies of integration of Rota-Baxter operators on Lie algebras. The differentiation of a Rota-Baxter Lie group give rise to a Rota-Baxter Lie algebra.
Then Bardakov and Gubarev studied the relationship between Rota-Baxter operators on groups and  skew-left braces, and constructed  solutions of the Yang-Baxter equation by means of Rota-Baxter operators on groups \mcite{BG22}. In fact, Rota-Baxter operators on groups naturally give rise to post-groups, while the latter are equivalent to skew-left braces as mentioned at the end of last subsection \cite{PostG}. This explained in another way the construction in \mcite{BG22}.

We summarize these structures and their applications in the following diagram.
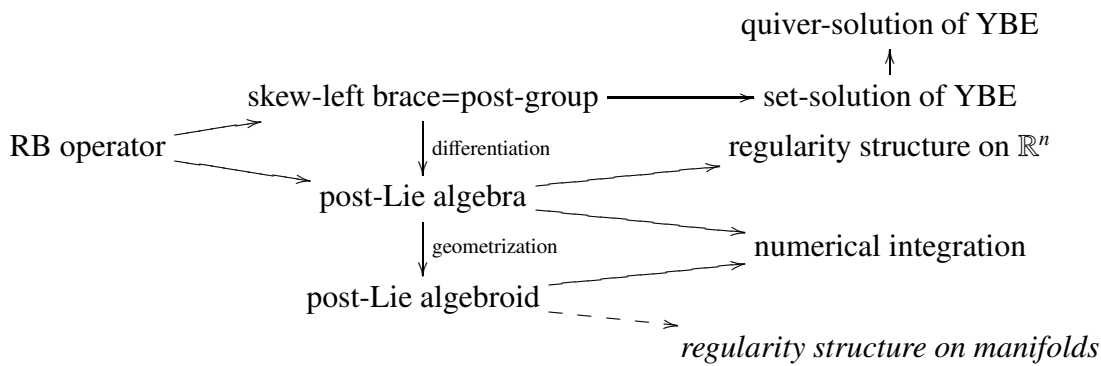
\begin{center}
 \hspace{3mm}\xymatrix@R=0.1pc{
 & & \text{quiver-solution of YBE} \\
 &  &\\
 & \text{skew-left brace=post-group} \ar[dd]^{\text{differentiation}}\ar[r]   & \text{set-solution of YBE} \ar[uu] \\
\text{RB operator} \ar[ur] \ar[dr]      & & \text{regularity structure on $\mathbb R^n$} \\
 & \text{post-Lie algebra} \ar[dd]^{\text{geometrization}} \ar[ur] \ar[dr] & \\
 &                                                                                      & \text{numerical integration} \\
 & \text{post-Lie algebroid} \ar[ur] \ar@{-->}[dr]                                    & \\
 &                                                                                      & \text{\em regularity structure on manifolds}
}
\end{center}

\subsection{Main results}

In this paper, we introduce the notion of post-groupoids, which consists of a group bundle and some other structures satisfying some compatibility conditions, and reduce to post-groups if the underlying base is one point. We show that a post-groupoid naturally gives rise to a groupoid (called the Grossman-Larson groupoid) and an action on the original group bundle.  An action of a group on a set naturally gives rise to a post-groupoid, whose Grossman-Larson groupoid is exactly the action groupoid. Moreover, we show that the section space of a post-groupoid is a weak post-group. This implies that the intrinsic structure of the observation in \cite{AFM} should be post-groupoids, as illustrated by the following diagram:
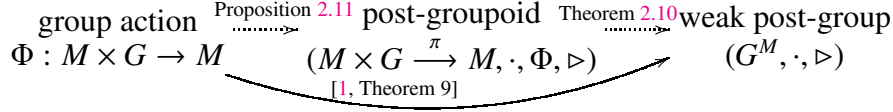
\begin{center}
\hspace{1mm}\xymatrix@R=0.1pc{
 \txt{group action \\ $\Phi:M\times G\to M$  }\ar@{..>}[r]^{\text{Proposition \ref{ex:post-g}}} \ar@/_{3pc}/[rr]!U^(.4){\text{\cite[Theorem 9]{AFM}}}& \txt{post-groupoid\\ $(M\times G\stackrel{\pi}{\longrightarrow}M,\cdot,\Phi,\rhd)$} \ar@{..>}[r]^{\text{Theorem \ref{thm:bisection}}}& \txt{weak post-group\\$(G^M, \cdot,\rhd)$} \\
 &  &
 }
 \end{center}
See Remark \ref{rmk:action-weak} for more explanation. In fact, we can enhance the results of \cite[Theorem 9]{AFM}, and obtain truly post-groups, by considering the space of ``bisections''  (Theorem \ref{thm:bisection-post-group}).

Note that  relative Rota-Baxter operators naturally induce post-structures both in the context of Lie algebras and groups. We introduce the notion of relative Rota-Baxter operators on a groupoid with respect to an action on a group bundle, and show that a relative Rota-Baxter operator naturally induces a post-groupoid. Conversely, a post-groupoid also gives rise to a relative Rota-Baxter operator on the corresponding Grossman-Larson groupoid. One important property of relative Rota-Baxter operators is that they give rise to   matched pairs of groupoids, which leads to the construction of quiver-theoretical solution of the Yang-Baxter equation on the quiver underlying the Grossman-Larson groupoid of a post-groupoid (Theorem \ref{pgybe}).

Motivated by the equivalence between post-groups and skew-left braces, we introduce the notion of skew-left bracoids, which consists of a group bundle and a groupoid satisfying some compatibility conditions and reduce to skew-left braces if the underlying base is one point. We show that associated to a post-groupoid, the original group bundle and the  Grossman-Larson groupoid form a skew-left bracoid naturally. Conversely, a skew-left bracoid also gives rise to a post-groupoid.

Finally, we consider post-Lie groupoids, that is, post-groupoids in the category of smooth manifolds. Using the fact that the  Grossman-Larson groupoid acts on the original group bundle, we show that one can obtain a post-Lie algebroid from a post-Lie groupoid via differentiation (Theorem \ref{thm:diff-post}).

We summarize the constructions and relations in the following diagram. The bold-faces terms and dotted arrows indicate new contribution we made in this article.

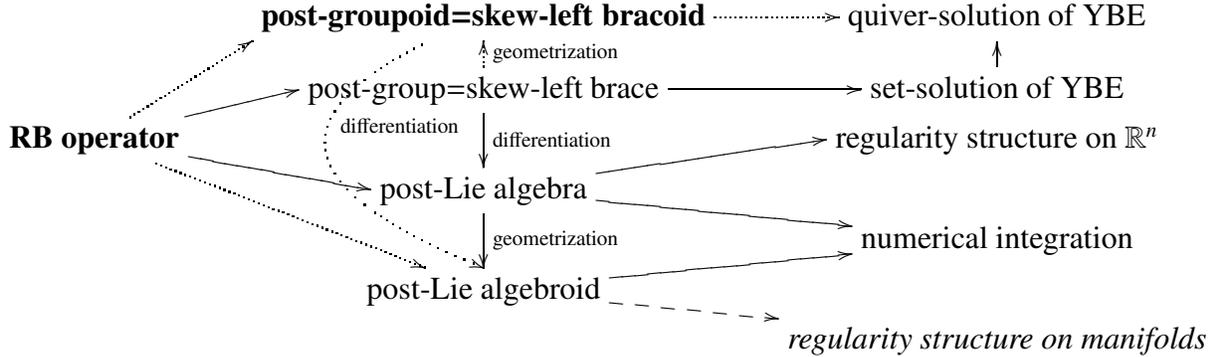
\xymatrix@R=0.1pc{
 & \text{\bf post-groupoid=skew-left bracoid} \ar@{..>}[r]  \ar@{..>}@/_{5pc}/[dddddd]!U^(.4){\text{differentiation}} & \text{quiver-solution of YBE} \\
 & & \\
 & \text{post-group=skew-left brace}\ar[dd]^{\text{differentiation}} \ar[r]+L \ar@{..>}[uu]_{\text{geometrization}}    & \text{set-solution of YBE} \ar[uu] \\
 \text{\bf RB operator} \ar[ur]+L \ar[dr]+L \ar@{..>}[uuur]+DL \ar@{..>}[dddr]    & & \text{regularity structure on $\mathbb R^n$} \\
 & \text{post-Lie algebra}  \ar[dd]^{\text{geometrization}} \ar[ur]+L \ar[dr] & \\
 &                                                                                      & \text{numerical integration} \\
 & \text{post-Lie algebroid} \ar[ur] \ar@{-->}[dr]                                    & \\
 &                                                                                      & \text{\em regularity structure on  manifolds}
}


\subsection{Future outlook}

To receive the  approximate solutions of the differential equations on the manifolds, Munthe-Kaas \cite{Munthe-Kaas} introduced the Lie-Butcher series  and studied  the Runge-Kutta method of the  differential equations on the manifolds. A closely related future direction  is to study the connection between integrators, as described by the Lie-Butcher series, and free post-Lie algebras whose universal enveloping algebras complete to universal Lie-Butcher series.  Let $G$ be a Lie group with the Lie algebra $\g=\Lie(G)$ and   $\rho:G\lon \Diff(M)$ be a transitive action of $G$ on a manifold $M$. Consider the differential equation on $M$: $
\frac{\dM u(t)}{\dM t}=F(u(t)),\,F\in\frkX(M).
$
 Using the  MKW post-Lie algebra $C^{\infty}(M,\g)$,
 one can rewrite the above differential equation by
$
\frac{\dM u(t)}{\dM t}=a_A\big(f\big(u(t)\big)\big),
$
for some $\,f\in C^{\infty}(M,\g)$. Note that the MKW post-Lie algebra structure on $C^{\infty}(M,\g)$ essentially comes form the MKW post-Lie algebroid $(A=M\times \g,[\cdot,\cdot]_\g,\triangleright_A,a_A)$ (as the section space), which is the differentiation of the MKW post-Lie groupoid $(M\times G\stackrel{\pi}{\longrightarrow}M,\cdot,\Phi,\rhd)$. Thus it is natural to explore differential equations on a general post-Lie groupoid $(\G\stackrel{\pi}{\longrightarrow}M,\cdot,\Phi,\rhd)$:\vspace{-2mm}
 \begin{eqnarray*}
\frac{\dM u(t)}{\dM t}=\Phi_*(f\big(u(t)\big)),
\end{eqnarray*}
here $\Phi_*$ is the anchor map of the post-Lie algebroid $(A\longrightarrow M,[\cdot,\cdot]_A,\Phi_*,\triangleright_A)$, which is the differentiation of the post-Lie groupoid $(\G\stackrel{\pi}{\longrightarrow}M,\cdot,\Phi,\rhd)$ and  $f\in\Gamma(A)$.
   The formal series expansion provided by the Lie-Butcher theory, which reflects the underlying post-Lie algebraic structure, could be investigated in the context of post-Lie algebroids. This extension may reveal new insights into the geometric and algebraic structures governing differential equations in these broader contexts. By exploring how integrators are linked to post-Lie algebroids and post-Lie groupoids via an analogous series expansion, one could potentially develop a unified framework for understanding and solving differential equations on more complex manifolds and spaces. 
   This could significantly advance the field and open up new avenues for applying post-Lie structures in mathematical physics, control theory, numerical analysis, and other more applied settings.


Built on the work \cite{HS}, where Hairer and Singh extended the theory of regularity structures for solving singular stochastic partial differential equations on Euclidean space $\RR^n$ or torus $\T^d$ to equations on vector bundles over a Riemannian manifold,   Singh  proposes the following question \cite{Singh}:\vspace{-1mm}
\begin{itemize}
  \item[$ \bullet$]  What is the  algebraic structure behind the regularity structures on manifolds and vector bundles?
\end{itemize}\vspace{-1.5mm}
Post-Lie algebras are the algebraic structure behind the regularity structures on $\RR^n$. Extending regularity structures on $\RR^n$ to manifolds and vector bundles involving suitable gluing. On the other hand, gluing post-Lie algebras together, one can obtain a post-Lie algebroid \cite{ML}. Thus it is very possible that the answer to Singh's question above is exactly a post-Lie algebroid or a post-Lie groupoid if we consider global symmetries. Therefore a key future direction is to understand how post-Lie algebroids and post-Lie groupoids manifest in the regular structures on manifolds and vector bundles, and to explicitly construct these structures.

Now a couple of problems which might attract attention in pure maths: in this article, we offer quiver-theoretical solutions to the Yang-Baxter equation using post-groupoids. We demonstrate that the category of post-groupoids is isomorphic to that of braided groupoids. A significant question that arises is how to utilize post-groupoids and their representation theory to classify non-degenerate solutions of the Yang-Baxter equation on quivers, parallel to what Andruskiewitsch did using braided groupoids \cite{Andruskiewitsch-1}. Furthermore, once the representation theory of post-groupoids is considered, a natural inquiry is the Morita equivalence of post-groupoids, as Morita equivalences lead to equivalences of categories of representations. This, in turn, opens up a wide range of further research themes involving other topological invariants, such as cohomology and homotopy groups, and how these invariants interact with the Morita equivalences of post-Lie groupoids.

Lastly, an intriguing challenge for geometers is to achieve a geometric integration of post-Lie algebras and post-Lie algebroids. While this article presents the differentiation from a post-Lie groupoid to a post-Lie algebroid, the reverse process integrating a post-Lie algebra into a post-Lie group remains elusive. Currently, the only known approach is a formal one, which employs a post-Lie version of the Baker-Campbell-Hausdorf formula   \cite{PostG}. However, finding a geometrically meaningful method for the integration of post-Lie algebras and post-Lie algebroids remains an open problem, posing an exciting direction for future research.


\section{Post-groupoids}\label{sec:2}

In this section, we introduce the notion of a post-groupoid. An action of a group on a set naturally gives rise to a post-groupoid. We show that on the section space and bisection space of a post-groupoid, there are weak post-group and post-group structures respectively. Moreover, a post-groupoid gives rise to a groupoid and an action on the original group bundle.

First we recall the notion of post-groups and their properties.

\begin{defi}\cite[Definition 2.1]{Wang}
A {\bf weak post-group} is a group $(G,\cdot)$ equipped with a multiplication $\rhd$ on $G$ such that
\begin{enumerate}
\item for each $a\in G$, the left multiplication $L^\rhd_a:G\to G, ~ L^\rhd_a b= a\rhd b$ for all $b\in G,$
is an endomorphism of the group $(G,\cdot)$, that is,
\begin{equation}
    \mlabel{Post-2}a\rhd (b\cdot c)=(a\rhd b)\cdot(a\rhd c)\tforall a, b, c\in G;
\end{equation}
\item the following condition holds,
    \begin{equation}
        \mlabel{Post-4}\big(a\cdot(a\rhd b)\big)\rhd c=a\rhd (b\rhd c)\tforall a, b, c\in G.
    \end{equation}
\end{enumerate}

The notion of a weak post-group is a generalization of the notion of a {\bf post-group} \cite[Definition 2.1]{PostG}, where the left multiplication $L^\rhd_a:G\to G$
is required to be an automorphism for all $a\in G$.
 \end{defi}

 \begin{thm}\mlabel{pro:subad}
Let $(G,\cdot,\rhd)$ be a post-group. Define $\star:G\times G\to G$     by
\begin{eqnarray}
\mlabel{eq:subad-com}a\star b&=&a\cdot(a\rhd b)\tforall a,b\in G.
\end{eqnarray}
\begin{enumerate}
    \item \mlabel{it:subad1}
Then $(G,\star)$ is a group with $e$ being the unit, and the inverse map $\dagger:G\to G$ given by
\begin{eqnarray}  \notag
 \mlabel{eq:subad-inv}a^\dagger&:=&(L^\rhd_a)^{-1}(a^{-1}) \tforall a\in G.
\end{eqnarray}
The   group $G_\rhd:=(G,\star)$ is called the {\bf  Grossman-Larson group} of the post-group $(G,\cdot,\rhd)$.
\item The left multiplication $L^\rhd:G\to \Aut(G)$ is an action of the group $(G,\star)$  on $(G,\cdot)$.
\mlabel{it:subad2}
 \end{enumerate}
\end{thm}

Recall that a groupoid~\cite{Mkz:GTGA} is a small category such that every morphism is invertible.

\begin{defi}
A {\bf groupoid} is a pair $(\G,M)$, where $M$ is the set of objects and $\G$ is the set of morphisms, with the  structure maps
\begin{itemize}
\item two surjective maps $\alpha,\beta: \G\longrightarrow M$, called the source map and target map, respectively;
\item  the multiplication $\cdot:\G_\beta\times_\alpha\G\longrightarrow \G$, where $\G_\beta\times_\alpha\G=\{(\gamma_1,\gamma_2)\in \G\times \G| \beta(\gamma_1)=\alpha(\gamma_2)\}$, such that
    $
\alpha(\gamma_1\cdot \gamma_2)=\alpha(\gamma_1),~\beta(\gamma_1\cdot \gamma_2)=\beta(\gamma_2),\tforall (\gamma_1,\gamma_2)\in\G_\beta\times_\alpha\G;
$
\item  the inverse map $\inv:\G\longrightarrow \G$;
\item the inclusion map $\iota: M\longrightarrow \G$, $m\longmapsto \iota_m$ called the identity map;
\end{itemize}
satisfying the following properties:
\begin{enumerate}
\item \rm{(associativity)} $(\gamma_1\cdot \gamma_2)\cdot \gamma_3=\gamma_1\cdot (\gamma_2\cdot \gamma_3)$, whenever the multiplications are well-defined;
\item \rm{(unitality)} $\iota_{\alpha(\gamma)}\cdot \gamma=\gamma=\gamma\cdot \iota_{\beta(\gamma)}$;
\item  \rm{(invertibility)} $\gamma\cdot \inv(\gamma)=\iota_{\alpha(\gamma)}$, $\inv(\gamma)\cdot \gamma=\iota_{\beta(\gamma)}$.
\end{enumerate}
We also denote a groupoid by $\xymatrix{ (\G \ar@<0.5ex>[r]^{\alpha} \ar[r]_{\beta} & M},\cdot,\iota,\inv)$ or simply by $\G$. A Lie groupoid is a groupoid such that both $\G$ and $M$ are smooth manifolds, $\alpha$ and $\beta$ are surjective submersion, and all the other structure maps are smooth.
 \end{defi}

\begin{ex}\label{ex:action-g}
Let $\Phi:M\times G\to M$ be a right action of a group $G$   on a set $M$. Then we obtain a groupoid $M\times G\rightrightarrows M$, whose source, target maps and multiplication are given by
\[\alpha(m,g):=m,\quad \beta(m,g):=\Phi(m,g),\quad (m,g)\star(n,h):=(m,g\cdot h),\]
for $n=\Phi(m,g)$. This  groupoid is called the {\bf action groupoid}.
\end{ex}

\begin{defi}\label{defi:action}
Let $\xymatrix{ (\G \ar@<0.5ex>[r]^{\alpha} \ar[r]_{\beta} & M},\cdot,\iota,\inv)$ be a groupoid, and $\pi:\HH\to M$ be a map of sets. A {\bf left action} of $\G$ on $\pi:\HH\to M$ is a map $\rightharpoonup:\G_\beta\times_\pi \HH\to \HH$ such that
\begin{equation}\label{laction}
\pi( \gamma\rightharpoonup\delta )=\alpha(\gamma),\quad \gamma_1\rightharpoonup (\gamma_2\rightharpoonup\delta)= (\gamma_1\cdot\gamma_2)\rightharpoonup\delta,\quad \iota_{\pi( \delta)}\rightharpoonup\delta =\delta.
\end{equation}
  \emptycomment{\begin{itemize}
    \item[(i)] $\pi( \gamma\rightharpoonup\delta )=\alpha(\gamma);$
    \item[(ii)]  For all $(\gamma_1, \gamma_2)\in\G_{\beta}\times _{\alpha}\G$, and $\delta\in\HH$ satisfying $\beta(\gamma_2)=\pi(\delta)$, we have $$ \gamma_1\rightharpoonup (\gamma_2\rightharpoonup\delta)= (\gamma_1\cdot\gamma_2)\rightharpoonup\delta;$$
        \item[(iii)] For all $\delta\in\HH$, $ \iota_{\pi( \delta)}\rightharpoonup\delta =\delta$.
  \end{itemize}
  }
  A {\bf right action} of $\G$ on $\pi:\HH\to M$ is a map $\leftharpoonup:\HH_\pi\times_\alpha \G\to \HH$ such that
 \begin{equation}\label{raction}
  \pi(\delta\leftharpoonup \gamma)=\beta(\gamma),\quad \delta\leftharpoonup (\gamma_1\cdot\gamma_2)=(\delta\leftharpoonup\gamma_1)\leftharpoonup\gamma_2,\quad \delta\leftharpoonup\iota_{\pi( \delta)} =\delta.
  \end{equation}
  \end{defi}

See \cite[Definition 2.5.1]{Mkz:GTGA} for the general notion of actions of groupoids on groupoids. Here we give the notion of actions of   groupoids on   group bundles, which suits for our applications. First we clarify terminologies. A bundle of groups is referred to be a groupoid in which the sauce map and the target map are the same. A group bundle is a bundle of groups $\HH\stackrel{\pi}{\longrightarrow}M$ such that {$\HH_m:=\pi^{-1}(m)$} are isomorphic for all $m\in M$. Denote by $\iota^\HH:M\to \HH$   the unit section. Denote by $\inv_\HH$ the inverse map. Denote a   group bundle by $(\HH\stackrel{\pi}{\longrightarrow}M,\cdot_\HH,\iota^\HH,\inv_\HH)$, or simply by $(\HH\stackrel{\pi}{\longrightarrow}M,\cdot,\iota,\inv)$, $\HH\stackrel{\pi}{\longrightarrow}M$ or $\HH$ if there is no risk of  confusion.

A bundle of Lie groups is  a Lie groupoid in which the sauce map and the target map are the same.
A Lie group bundle is a bundle of Lie groups $\HH\stackrel{\pi}{\longrightarrow}M$ such that for every $m\in M$, there is an open set $U$ containing $m$, a   group $G$, and a homeomorphism $\psi:U\times G\to \pi^{-1}(U)$ such that $\psi_m:m\times G\to \HH_m=\pi^{-1}(m)$ is a   group isomorphism \cite[Definition 1.1.19]{Mkz:GTGA}.

\begin{defi}\label{defi:action-g}
 An action of a groupoid $\xymatrix{ (\G \ar@<0.5ex>[r]^{\alpha} \ar[r]_{\beta} & M},\cdot_\G,\iota^\G,\inv_\G)$ on a  group bundle $\HH\stackrel{\pi}{\longrightarrow}M$ is a   map $\rightharpoonup:\G_{\beta}\times_{\pi}\HH\to \HH$ satisfying  \eqref{laction} and such that
   the map $ \gamma\rightharpoonup\cdot:\HH_{\beta(\gamma)}\to \HH_{\alpha(\gamma)}$ is an isomorphism of  groups for all  $\gamma\in\G$.
\end{defi}

Recall that for a group bundle $\HH\stackrel{\pi}{\longrightarrow}M$, there is an associated groupoid $\Aut(\HH)$, whose space of objects is $M$, and morphisms from $m$ to $n$ is isomorphisms of groups from $\HH_n$ to $\HH_m$. Then an action of a groupoid $\xymatrix{ (\G \ar@<0.5ex>[r]^{\alpha} \ar[r]_{\beta} & M},\cdot_\G,\iota^\G,\inv_\G)$ on a group bundle $\HH\stackrel{\pi}{\longrightarrow}M$ is equivalent to a groupoid homomorphism from the groupoid $\xymatrix{ (\G \ar@<0.5ex>[r]^{\alpha} \ar[r]_{\beta} & M},\cdot_\G,\iota^\G,\inv_\G)$ to the groupoid $\Aut(\HH)$.

Now we are ready to introduce the notion of a post-groupoid, which is the main object studied in this paper.
\begin{defi}\label{defi:post-groupoid}
  A {\bf post-(Lie) groupoid} consists of the following data:
  \begin{itemize}
    \item a (Lie) group bundle $(\G\stackrel{\pi}{\longrightarrow}M,\cdot,\iota,\inv)$, where $\cdot$ is the multiplication;
    \item a surjective map (submersion) $\Phi:\G\to M$ satisfying $\Phi(\iota_m)=m,$ for all $m\in M$;
    \item a (smooth) map $\rhd: \G_\Phi\times_\pi\G\to \G$ satisfying $\pi(\gamma\rhd \delta)=\pi(\gamma)$ and the left multiplication $$L^\rhd_\gamma:\G_{\Phi(\gamma)}\to \G_{\pi(\gamma)}, \quad L^\rhd_\gamma \delta= \gamma\rhd \delta \tforall \delta\in \G_{\Phi(\gamma)},$$
is invertible for any $\gamma\in\G$, where $\G_\Phi\times_\pi \G=\{(\gamma,\delta)\in \G\times \G ~\mbox{such that}~ \Phi(\gamma)=\pi(\delta)\},$
  \end{itemize}
  such that for all $(\gamma_1,\gamma_2,\gamma_3)\in \G_\Phi\times_\pi \G_\Phi\times_\pi \G$, and $\gamma_2'\in \G$ satisfying $\pi(\gamma_2)=\pi(\gamma_2')$, the following axioms hold:
    \begin{itemize}
    \item[\rm(i)]  $\Phi\Big(\gamma_1\cdot(\gamma_1\rhd\gamma_2)\Big)=\Phi(\gamma_2) $,
     \item[\rm(ii)]  $\gamma_1\rhd(\gamma_2\cdot \gamma_2')=(\gamma_1\rhd\gamma_2)\cdot(\gamma_1\rhd\gamma_2')$,
      \item[\rm(iii)] $\gamma_1\rhd(\gamma_2\rhd \gamma_3)=\Big(\gamma_1\cdot(\gamma_1\rhd\gamma_2)\Big) \rhd\gamma_3$,
  \end{itemize}
  where $\G_\Phi\times_\pi \G_\Phi\times_\pi \G=\{(\gamma_1,\gamma_2,\gamma_3)\in \G\times \G\times \G ~\mbox{such that}~ \Phi(\gamma_1)=\pi(\gamma_2),~ \Phi(\gamma_2)=\pi(\gamma_3)\}.$
  We will denote a post-(Lie) groupoid by $(\G\stackrel{\pi}{\longrightarrow}M,\cdot,\Phi,\rhd)$.
\end{defi}

\begin{defi}\label{defi:post-groupoid-mor}
   Let  $(\G\stackrel{\pi_\G}{\longrightarrow}M,\cdot_\G,\Phi_\G,\rhd_\G)$ and  $(\HH\stackrel{\pi_\HH}{\longrightarrow}M,\cdot_\HH,\Phi_\HH,\rhd_\HH)$ be post-groupoids. A (base-preserving) {\bf post-groupoid homomorphism} is a map $\Psi:\G\to\HH$ such that
    \begin{itemize}
    \item[\rm(i)]  $\pi_\HH\circ \Psi=\pi_\G,\quad  \Phi_\HH\circ \Psi=\Phi_\G$,

     \item[\rm(ii)]  $\Psi(\gamma \cdot_\G \gamma')=\Psi(\gamma)\cdot_\HH\Psi(\gamma'),\quad \mbox{for all}~ \gamma, \gamma'\in \G~\mbox{satisfying }~\pi_\G(\gamma)=\pi_\G(\gamma')$,

      \item[\rm(iii)] $\Psi(\gamma \rhd_\G \delta)=\Psi(\gamma)\rhd_\HH\Psi(\delta),\quad \mbox{for all}~ (\gamma, \delta)\in \G_{\Phi_\G}\times_{\pi_\G}\G$.
  \end{itemize}
\end{defi}

Post-groupoids and their homomorphisms form a category, denoted by $\PG$.

\delete{Parallel to the fact that a groupoid has an isotropic group at each point, a post-groupoid has an isotropic post-group at each point. Let $(\G\stackrel{\pi}{\longrightarrow}M,\cdot,\Phi,\rhd)$ be a post-groupoid. For all $m\in M$, Define $\G^m_m$ by
$$
G^m_m=\{\gamma\in\G|~\pi(\gamma)=\Phi(\gamma)=m\}.
$$
}
\begin{lem}\label{lem:unit}
Let $(\G\stackrel{\pi}{\longrightarrow}M,\cdot,\Phi,\rhd)$ be a post-groupoid. Then for all $\gamma\in\G$, we have
\begin{eqnarray*}
\gamma\rhd\iota_{\Phi(\gamma)}=\iota_{\pi(\gamma)},\quad
\iota_{\pi(\gamma)}\rhd\gamma=\gamma.
\end{eqnarray*}

\end{lem}

\begin{proof}
  By   Axiom (ii) in Definition \ref{defi:post-groupoid}, we have
  $
  \gamma\rhd\iota_{\Phi(\gamma)}=\gamma\rhd (\iota_{\Phi(\gamma)}\cdot\iota_{\Phi(\gamma)})=(\gamma\rhd \iota_{\Phi(\gamma)})\cdot(\gamma\rhd\iota_{\Phi(\gamma)}).
  $
  Since $\gamma\rhd\iota_{\Phi(\gamma)}\in\G_{\pi(\gamma)}$, it follows that $\gamma\rhd\iota_{\Phi(\gamma)}=\iota_{\pi(\gamma)}$.

By  Axiom (iii) in Definition \ref{defi:post-groupoid}, we have
  $
 \iota_{\pi(\gamma)}\rhd(\iota_{\pi(\gamma)}\rhd\gamma)= (\iota_{\pi(\gamma)}\cdot(\iota_{\pi(\gamma)}\rhd\iota_{\pi(\gamma)}))\rhd\gamma=\iota_{\pi(\gamma)}\rhd\gamma.
  $
  Since $L^\rhd_{\iota_{\pi(\gamma)}}$ is invertible, it follows that $\iota_{\pi(\gamma)}\rhd\gamma=\gamma.$
\end{proof}

Denote by $\Gamma(\G)$ the set of sections of the group bundle $\G$, i.e.
$$
\Gamma(G)=\{\sigma:M\to \G|~\pi\circ \sigma=\Id_M\}.
$$
Then the  multiplication $\cdot$ and the map $\rhd$ can naturally be extended on $\Gamma(\G)$, for which we use the same notations, via the following formulas:
\begin{eqnarray}
  \label{formula:weak1}(\sigma_1\cdot\sigma_2) (m)&=&\sigma_1(m)\cdot\sigma_2(m),\\
    \label{formula:weak2}(\sigma_1\rhd\sigma_2) (m)&=&\sigma_1(m)\rhd \sigma_2(\Phi(\sigma_1(m))).
\end{eqnarray}

We have the following result.

\begin{thm}\label{thm:bisection}
  Let $(\G\stackrel{\pi}{\longrightarrow}M,\cdot,\Phi,\rhd)$ be a post-groupoid. Then $(\Gamma(\G), \cdot,\rhd)$ is a weak post-group.
\end{thm}

\begin{proof}
  It is obvious that $(\Gamma(\G),\cdot)$ is a group whose identity is the unit section $\iota:M\to\G$.


For all $\sigma_1,\sigma_2,\sigma_3\in \Gamma(\G)$, by Axiom (ii) in Definition \ref{defi:post-groupoid}, we have
\begin{eqnarray*}
  \sigma_1\rhd(\sigma_2\cdot \sigma_3)(m)&=&\sigma_1(m)\rhd (\sigma_2\cdot\sigma_3)(\Phi(\sigma_1(m)))\\
  &=&\sigma_1(m)\rhd (\sigma_2(\Phi(\sigma_1(m)))\cdot\sigma_3(\Phi(\sigma_1(m))))\\
  &=&(\sigma_1(m)\rhd  \sigma_2(\Phi(\sigma_1(m))))\cdot(\sigma_1(m)\rhd\sigma_3(\Phi(\sigma_1(m))))\\
  &=&(\sigma_1\rhd \sigma_2)(m)\cdot(\sigma_1\rhd \sigma_3)(m)\\
  &=&\Big((\sigma_1\rhd \sigma_2)\cdot(\sigma_1\rhd \sigma_3)\Big)(m),
\end{eqnarray*}
which implies that $\sigma_1\rhd(\sigma_2\cdot \sigma_3)=(\sigma_1\rhd \sigma_2)\cdot(\sigma_1\rhd \sigma_3).$

Similarly, by Axiom (i) and (iii) in Definition \ref{defi:post-groupoid}, we have
\begin{eqnarray*}
  \sigma_1\rhd(\sigma_2\rhd \sigma_3)(m)&=&\sigma_1(m)\rhd (\sigma_2\rhd\sigma_3)(\Phi(\sigma_1(m)))\\
  &=&\sigma_1(m)\rhd \Big(\sigma_2(\Phi(\sigma_1(m)))\rhd\sigma_3(\Phi(\sigma_2(\Phi(\sigma_1(m)))))\Big)\\
  &=&\Big(\sigma_1(m)\cdot(\sigma_1(m)\rhd  \sigma_2(\Phi(\sigma_1(m))))\Big)\rhd\sigma_3(\Phi(\sigma_2(\Phi(\sigma_1(m)))))\\
  &=&\Big(\sigma_1\cdot(\sigma_1\rhd \sigma_2)\Big)(m) \rhd \sigma_3(\Phi(\sigma_2(\Phi(\sigma_1(m)))))\\
  &=&\Big(\sigma_1\cdot(\sigma_1\rhd \sigma_2)\Big) \rhd \sigma_3 (m),
\end{eqnarray*}
which implies that $\sigma_1\rhd(\sigma_2\rhd \sigma_3)=\Big(\sigma_1\cdot(\sigma_1\rhd \sigma_2)\Big) \rhd \sigma_3.$

Therefore, $( \Gamma(\G), \cdot,\rhd)$ is a weak post-group.
\end{proof}

  Let $G$ be a group and $M$ a set. Then $\G=M\times G\stackrel{\pi}{\longrightarrow}M$ is  a trivial group bundle, where $\pi$ is the projection to $M$. Let $\Phi:M\times G\to M$ be a right action of $G$ on $M$. In this case, $$\G_\Phi\times_\pi\G=\{((m,g),(n,h))\in \G\times \G ~\mbox{such that}~ n=\Phi(m,g)\}.$$

\begin{pro}\label{ex:post-g}
 $\Phi:M\times G\to M$ be a right action of $G$ on $M$. Then $(M\times G\stackrel{\pi}{\longrightarrow}M,\cdot,\Phi,\rhd)$ is a post-groupoid, where   $\rhd:\G_\Phi\times_\pi\G\to \G$ is defined  by
 \begin{equation}\label{fomula:post-action}
  (m,g)\rhd(n,h)=(m,h).
\end{equation}
\end{pro}
\begin{proof}
  It is obvious that $\pi((m,g)\rhd(n,h))=\pi(m,h)=m=\pi(m,g)$, and $(m,g)\rhd(n,h)=(m, e_m)$ if and only if $h=e_n$, which implies that $L^\rhd_{(m,g)}:\G_{\Phi(m,g)}\to\G_m$ is invertible.

  For $(m,g),(n,h)\in M\times G$ satisfying $n=\Phi(m,g)$, we have
  \begin{eqnarray*}
    (m,g)\cdot((m,g)\rhd(n,h))=(m,g)\cdot(m,h)=(m,g\cdot h),
  \end{eqnarray*}
  which implies that $$\Phi((m,g)\cdot((m,g)\rhd(n,h)))=\Phi(m,g\cdot h)=\Phi(\Phi(m,g),h)=\Phi(n,h).$$
   So Axiom (i)   in Definition \ref{defi:post-groupoid} holds.

  For all $(m,g),(n,h),(n,h')\in M\times G$ satisfying $n=\Phi(m,g)$, we have
  \begin{eqnarray*}
    (m,g)\rhd((n,h)\cdot(n,h'))&=&(m,g)\rhd(n,h\cdot h')=(m,h\cdot  h'),\\
    ((m,g)\rhd(n,h))\cdot((m,g)\rhd(n,h'))&=&(m,h)\cdot(m,h')=(m,h\cdot h'),
  \end{eqnarray*}
  which implies that  Axiom (ii)   in Definition \ref{defi:post-groupoid} holds.

  Finally, for $(m,g),(n,h),(p,k)\in M\times G$ satisfying $n=\Phi(m,g)$ and $p=\Phi(n,h)$, we have
    \begin{eqnarray*}
    (m,g)\rhd( (n,h)\rhd (p,k))&=&(m,g)\rhd(n,k)=(m,k),\\
    \Big((m,g)\cdot((m,g)\rhd(n,h))\Big)\rhd(p,k)&=&(m,g\cdot h)\rhd(p,k)=(m,k),
  \end{eqnarray*}
  which implies that Axiom (iii)   in Definition \ref{defi:post-groupoid} holds.

  Therefore, $(M\times G\stackrel{\pi}{\longrightarrow}M,\cdot,\Phi,\rhd)$ is a post-groupoid.
  \end{proof}

\begin{rmk}
  If $G$ is a Lie group and $M$ is a manifold, then we obtain a post-Lie groupoid $(M\times G\stackrel{\pi}{\longrightarrow}M,\cdot,\Phi,\rhd)$, which is called the {\bf MKW post-Lie groupoid} and play important roles in the study of numerical Lie-group integration on a homogeneous space \cite[ Section 2.2.1]{ML}.
\end{rmk}

\begin{rmk}\label{rmk:action-weak}
  Under the assumption of Proposition \ref{ex:post-g}, a weak post-group is constructed in \cite[Theorem 9]{AFM}. More precisely, consider  the space of $G$-valued functions on $M$, which is denoted by $G^M$, and endowed with the pointwise product
\begin{equation}\label{formula:fard1}
 ( f\cdot g)(m)=f(m)\cdot g(m),\quad \forall f,g\in G^M.
\end{equation} Introduce the map $\rhd: G^M \times G^M \to G^M$  by
\begin{equation}\label{formula:fard2}
(f \rhd g)(m) := g (\Phi(m,f (m))).
\end{equation}
Then $(G^M,\cdot,\rhd)$ is a weak post-group.

On the other hand, $G$-valued functions one-to-one correspond to sections of the group bundle $\G=M\times G\stackrel{\pi}{\longrightarrow}M$ via
\begin{equation}\label{eq:corr}
f\in G^M\longleftrightarrow \sigma_f:M\to M\times G, \quad \sigma_f(m):=(m, f(m)).
\end{equation}
So there is an induced weak post-group structure on the set of sections of the group bundle $\G=M\times G\stackrel{\pi}{\longrightarrow}M$ given by
\begin{eqnarray*}
  (\sigma_f\cdot\sigma_g)(m):=(m,f\cdot g(m)),\quad
  (\sigma_f\rhd\sigma_g)(m):=(m,f\rhd g(m))=(m,g (\Phi(m,f (m)))).
\end{eqnarray*}
 By \eqref{fomula:post-action}, we have
 $
 (m,g (\Phi(m,f (m))))=\sigma_f(m)\rhd \sigma_g(\Phi(\sigma_f(m))),
 $
 which implies that this weak post-group is exactly the one
   given by Theorem \ref{thm:bisection} for the post-groupoid $(M\times G\stackrel{\pi}{\longrightarrow}M,\cdot,\Phi,\rhd)$.
\end{rmk}

 In fact, there is indeed post-groups underlying group actions, or more generally, there is post-groups associated to post-groupoids. Consider the subset $\Gamma_b(\G)\subset \Gamma(\G)$, which consists of those sections $\sigma $ satisfying $\Phi\circ \sigma$ is a diffeomorphism on $M$.  As an enhancement of Theorem \ref{thm:bisection}, we have the following result.
\begin{thm}\label{thm:bisection-post-group}
  Let $(\G\stackrel{\pi}{\longrightarrow}M,\cdot,\Phi,\rhd)$ be a post-groupoid. Then $(\Gamma_b(\G), \cdot,\rhd)$ is a post-group, where $\cdot$ and $\rhd$ are defined by \eqref{formula:weak1} and \eqref{formula:weak2} respectively.  Consequently, $(\Gamma_b(\G),\star)$ is a group, where the group multiplication $\star$ is given by
   $   \sigma\star \tau =\sigma \cdot (\sigma\rhd \tau).    $ More precisely,
   $$
    (\sigma\star \tau) (m)=\sigma (m)\cdot( \sigma (m)\rhd \tau(\Phi(\sigma(m)))),\quad \forall \sigma,\tau\in \Gamma_b(\G), m\in M.
   $$
\end{thm}
\begin{proof}
  Obviously, we only need to show that in this case $L^\rhd_\sigma:\Gamma_b(\G)\to \Gamma_b(\G)$ is invertible for all $\sigma\in \Gamma_b(\G)$. In fact, if $\sigma\rhd \tau=\iota$, we have
  $$
  (\sigma\rhd \tau) (m)=\sigma(m)\rhd\tau(\Phi(\sigma(m)))=\iota_m,\quad \forall m\in M.
  $$
  Since $L^\rhd_{\sigma(m)}:\G_{\Phi(\sigma(m))}\to \G_m$ is invertible, it follows that $\tau(\Phi(\sigma(m)))=\iota_{\Phi(\sigma(m))}$ for all $m\in M$. Since $\Phi\circ\sigma$ is diffeomorphism, this implies that $\tau(m)=\iota_m$ for all $m\in M$. Therefore, $L^\rhd_\sigma:\Gamma_b(\G)\to \Gamma_b(\G)$ is injective. Similarly, it is also surjective. Thus,  $(\Gamma_b(\G), \cdot,\rhd)$ is a post-group.
\end{proof}
Consider the   post-groupoid given in Proposition \ref{ex:post-g}, we obtain the following post-group.
\begin{cor}\label{cor:subadj}
   Let $\Phi:M\times G\to M$ be a right action of $G$ on $M$. Define the set $G^M_\Phi$ by
   $$
   G^M_\Phi:=\{f\in G^M ~\mbox{such that}~\Phi\circ f ~\mbox{is a diffeomorphism on $M$}\}.
   $$
   Then $(G^M_\Phi,\cdot,\rhd)$ is a post-group, where $\cdot$ and $\rhd$ are given by \eqref{formula:fard1} and \eqref{formula:fard2}. Consequently, $(G^M_\Phi,\star)$ is a group, where the  multiplication $\star$ is given by
   $   f\star g =f \cdot (f\rhd g).    $ More precisely,
   $$
    (f\star g) (m)=f (m)\cdot  g(\Phi(m, f(m))),\quad \forall f,g\in G^M_\Phi, m\in M.
   $$
\end{cor}
\begin{proof}
  It follows from Theorem \ref{thm:bisection-post-group}, and the correspondence between $G$-valued functions and sections given by \eqref{eq:corr} directly.
\end{proof}
In the sequel, we show that a post-groupoid gives rise to a new groupoid naturally, and the set $\Gamma_b(\G)$ can be understood as the set of bisections.
\begin{thm}\label{thm:subadj}
 Let $(\G\stackrel{\pi}{\longrightarrow}M,\cdot,\Phi,\rhd)$ be a post-groupoid.   Then $(\xymatrix{ \G \ar@<0.5ex>[r]^{\alpha=\pi} \ar[r]_{\beta=\Phi} & M},\star,\iota,\inv_\rhd)$ is a groupoid, called the {\bf Grossman-Larson groupoid}, where the   groupoid multiplication $\star:\G_\Phi\times_\pi\G\to\G$ is defined by
 \begin{equation}
   \gamma_1\star\gamma_2=\gamma_1\cdot(\gamma_1\rhd\gamma_2),
 \end{equation}
 and the new inverse map $\inv_\rhd$ is given by
 \begin{equation}
   \inv_\rhd(\gamma)=(L^\rhd_\gamma)^{-1}\inv(\gamma).
 \end{equation}

Moreover, $\rhd$ gives rise to an action of the Grossman-Larson groupoid  $(\xymatrix{ \G \ar@<0.5ex>[r]^{\alpha=\pi} \ar[r]_{\beta=\Phi} & M},\star,\iota,\inv_\rhd)$ on the group bundle $(\G\stackrel{\pi}{\longrightarrow}M,\cdot,\iota,\inv)$.
\end{thm}

\begin{proof}
Similar to the proof of \cite[Theorem 2.4]{PostG}, for all $(\gamma_1,\gamma_2,\gamma_3)\in\G_\Phi\times_\pi\G_\Phi\times_\pi\G$, we have
$$
(\gamma_1\star\gamma_2)\star\gamma_3=\gamma_1\star(\gamma_2\star\gamma_3),
$$
which implies that $\star$ is associative.

By Lemma \ref{lem:unit}, we have
\begin{eqnarray*}
 \gamma\star \iota_{\Phi(\gamma)}=\gamma\cdot(\gamma\rhd\iota_{\Phi(\gamma)})=\gamma\cdot\iota_{\pi(\gamma)}=\gamma,\quad
\iota_{\pi(\gamma)}\star\gamma=\iota_{\pi(\gamma)}\cdot(\iota_{\pi(\gamma)}\rhd\gamma)=\iota_{\pi(\gamma)}\cdot\gamma=\gamma,
\end{eqnarray*}
which implies that $\iota$ is the identity map.

Finally, we have
\begin{eqnarray*}
  \gamma\star  \inv_\rhd(\gamma)=\gamma\star(L^\rhd_\gamma)^{-1}\inv(\gamma)=\gamma\cdot(L^\rhd_\gamma(L^\rhd_\gamma)^{-1}\inv(\gamma))=\gamma\cdot\inv(\gamma)=\iota_{\pi(\gamma)}.
 \end{eqnarray*}
By Axiom (ii) and Axiom (iii)   in Definition \ref{defi:post-groupoid} and Lemma \ref{lem:unit}, we have
\begin{eqnarray*}
  \gamma\rhd( \inv_\rhd(\gamma)\star\gamma)&=& \gamma\rhd\Big((L^\rhd_\gamma)^{-1}\inv(\gamma)\cdot((L^\rhd_\gamma)^{-1}\inv(\gamma)\rhd\gamma)\Big)\\
  &=&\inv(\gamma)\cdot\Big(\gamma\rhd((L^\rhd_\gamma)^{-1}\inv(\gamma)\rhd\gamma)\Big)\\
  &=&\inv(\gamma)\cdot\Big((\gamma\cdot (\gamma\rhd(L^\rhd_\gamma)^{-1}\inv(\gamma)))\rhd\gamma\Big)\\
  &=&\inv(\gamma)\cdot\Big((\gamma\cdot \inv(\gamma))\rhd\gamma\Big)=\inv(\gamma)\cdot\Big(\iota_{\pi(\gamma)}\rhd\gamma\Big)=\inv(\gamma)\cdot\gamma=\iota_{\pi(\gamma)}.
\end{eqnarray*}
Since $L^\rhd_\gamma:\G_{\Phi(\gamma)}\to \G_{\pi(\gamma)}$ is invertible, it follows that $ \inv_\rhd(\gamma)\star\gamma=\iota_{\Phi(\gamma)}$. So $ \inv_\rhd(\gamma)$ is the inverse of $\gamma$ with respect to the multiplication $\star.$
Therefore,  $\xymatrix{ (\G \ar@<0.5ex>[r]^{\alpha=\pi} \ar[r]_{\beta=\Phi} & M},\star,\iota,\inv_\rhd)$ is a groupoid.
 It is straightforward to check that $\rhd$ gives rise to an action of the groupoid $\xymatrix{ (\G \ar@<0.5ex>[r]^{\alpha=\pi} \ar[r]_{\beta=\Phi} & M},\star,\iota,\inv_\rhd)$ on the group bundle $(\G\stackrel{\pi}{\longrightarrow}M,\cdot,\iota,\inv)$.
\end{proof}

\begin{rmk}
  The space $\Gamma_b(\G)$ considered in Theorem \ref{thm:bisection-post-group} is exactly the set of bisections of the Grossman-Larson groupoid $\xymatrix{ (\G \ar@<0.5ex>[r]^{\alpha=\pi} \ar[r]_{\beta=\Phi} & M},\star,\iota,\inv_\rhd)$. Furthermore, the group $(\Gamma_b(\G),\star)$ is exactly the group of bisections associated to the  groupoid $\xymatrix{ (\G \ar@<0.5ex>[r]^{\alpha=\pi} \ar[r]_{\beta=\Phi} & M},\star,\iota,\inv_\rhd)$.
\end{rmk}

\begin{pro}\label{pro:subadj-mor}
 Let $\Psi:\G\to\HH$ be  a homomorphism from a post-groupoid $(\G\stackrel{\pi_\G}{\longrightarrow}M,\cdot_\G,\Phi_\G,\rhd_\G)$ to  $(\HH\stackrel{\pi_\HH}{\longrightarrow}M,\cdot_\HH,\Phi_\HH,\rhd_\HH)$. Then   $\Psi:\G\to\HH$ is a (base-preserving) homomorphism between the corresponding Grossman-Larson groupoids.
\end{pro}
\begin{proof}
 Condition (i) in Definition \ref{defi:post-groupoid-mor} implies that $\Psi$ preserves the sauce map and the target map in the corresponding Grossman-Larson groupoids.  Condition (ii) and Condition (iii) in Definition \ref{defi:post-groupoid-mor} implies that $\Psi$ preserves the multiplication $\star$. Thus $\Psi$ is also a groupoid homomorphism.
\end{proof}

\begin{ex}
  Consider the Grossman-Larson groupoid of the post-groupoid given in Proposition \ref{ex:post-g}. Then for $ (m,g),~  (n,h)\in M\times G$ satisfying $n=\Phi(m,g)$, we have
  $$
  (m,g)\star (n,h)=(m,g)\cdot((m,g)\rhd(n,h))=(m,g)\cdot(m,h)=(m, g\cdot h).
  $$
  So the Grossman-Larson groupoid of the post-groupoid  $(M\times G\stackrel{\pi}{\longrightarrow}M,\Phi,\rhd)$ given in Proposition \ref{ex:post-g} is exactly the action groupoid given in Example \ref{ex:action-g}.  The space $G^M_\Phi$ considered in Corollary \ref{cor:subadj} is exactly the set of bisections of the action groupoid.
\end{ex}

At the end of this section,  we introduce the notion of an action of a post-group on a set, and show that an action of a post-group on a set gives rise to a post-groupoid.

\begin{defi}
  A {\bf right action of a post-group} $(G,\cdot,\rhd)$ on a set $M$ is defined to be a right action $\Phi:M\times G\to M$ of the Grossman-Larson group $(G,\star)$ on the set $M$.
\end{defi}

Let  $\Phi:M\times G\to M$ be a right action of a post-group $(G,\cdot,\rhd)$  on a set $M$. Then $\G=M\times G\stackrel{\pi}{\longrightarrow}M$ is  a trivial group bundle, where $\pi$ is the projection to $M$. Furthermore, the post-product $\rhd$ naturally extends to a  map $\bar{\rhd}: \G_\Phi\times_\pi\G\to \G$, via
$$
(m,g)\bar{\rhd}(n,h)=(m,g\rhd h), \quad\forall (m,g),~(n,h)\in M\times G ~\mbox{satisfying}~n=\Phi(m,g).
$$
Obviously, there holds $\pi(m,g)=\pi((m,g)\bar{\rhd}(n,h))$.

\begin{pro}\label{pro:act-post-groupoid}
 Let  $\Phi:M\times G\to M$ be a right action of a post-group $(G,\cdot,\rhd)$  on a set $M$. Then  $(M\times G\stackrel{\pi}{\longrightarrow}M,\cdot,\Phi,\bar{\rhd})$ is a post-groupoid, whose Grossman-Larson groupoid is exactly the action  groupoid of the Grossman-Larson group $(G,\star)$ on $M$, i.e. we have the following commutative diagram:
 \begin{equation}
\begin{split}
\xymatrix{
	\text{actions of post-groups on } M \ar^{\text{\qquad action}}[rr]   \ar@{=}[d] && \text{action post-groupoids}  \ar^{\text{Grossman-Larson groupoid}}[d]\\
	\text{ actions of Grossman-Larson groups on } M \ar^{\text{\qquad action}}[rr] && \text{action groupoids.}
}
\end{split}
\end{equation}	
\end{pro}

\begin{proof}
  First the requirement that $L^{\bar{\rhd}}_{(m,g)}:\G_{\Phi(m,g)}\to \G_{\pi(m,g)}$ is invertible follows from the fact that $L^{  {\rhd}}_{g}:G\to G$ is invertible.

  For all $ (m,g),~  (n,h)\in M\times G$ satisfying $n=\Phi(m,g)$, we have
 \begin{eqnarray*}
   \Phi\Big((m,g)\cdot((m,g)\bar{\rhd}(n,h))\Big)
  =\Phi\Big((m,g\cdot(g\rhd h))\Big)=\Phi(m,g\star h)=\Phi(\Phi(m,g),h)=\Phi(n,h),
 \end{eqnarray*}
 which implies that Axiom (i)   in Definition \ref{defi:post-groupoid} holds.

 Axiom (ii) and Axiom (iii)  in Definition \ref{defi:post-groupoid} also follows from \eqref{Post-2} and \eqref{Post-4} directly. Thus, $(M\times G\stackrel{\pi}{\longrightarrow}M,\cdot,\Phi,\bar{\rhd})$ is a post-groupoid.

The other conclusions are straightforward to be verified.
\end{proof}

 If  $(G,\cdot,\rhd)$ is a post-Lie group and $M$ is a smooth manifold, then we obtain a post-Lie groupoid $(M\times G\stackrel{\pi}{\longrightarrow}M,\cdot,\Phi,\bar{\rhd})$

\begin{rmk}
  An action of the Grossman-Larson group of a post-group gives rise to two post-groupoids via Proposition  \ref{pro:act-post-groupoid} and Proposition \ref{ex:post-g} respectively. Even though these two post-groupoids are different, but they have the same Grossman-Larson groupoid: the action  groupoid.
\end{rmk}

\section{Post-groupoids and relative Rota-Baxter operators}\label{sec:3}

In this section, we introduce the notion of relative Rota-Baxter operators on a groupoid with respect to an action, and establish the connection with post-groupoids. In particular, a relative Rota-Baxter operator naturally induces a post-groupoid. Moreover, a relative Rota-Baxter operator also gives rise to a matched pair of groupoids, which plays important roles in our later study of quiver-theoretical solutions of the Yang-Baxter equation.

\begin{defi}
   Let $ \rightharpoonup:\G_{\beta}\times_{\pi}\HH\to \HH$ be a left action of a  groupoid $\xymatrix{ (\G \ar@<0.5ex>[r]^{\alpha} \ar[r]_{\beta} & M},\cdot_\G,\iota^\G,\inv_\G)$ on a group bundle $(\HH\stackrel{\pi}{\longrightarrow}M,\cdot_\HH,\iota^{\HH},\inv_{\HH})$.  A smooth map $\B:\HH\to\G$ is called a {\bf relative Rota-Baxter operator} on $\G$ with respect to the action $\rightharpoonup$ if
   \begin{eqnarray}
  \label{eq:s-p}\alpha(\B(\delta))&=&\pi(\delta),\\
  \label{eq:unit1}\B(\iota^\HH_m)&=&\iota^\G_m,\\
  \label{eq:rRB}
    \B(\delta_1)\cdot_\G\B(\delta_2)&=&\B\Big(\delta_1\cdot_\HH (\B(\delta_1)\rightharpoonup\delta_2)\Big),
  \end{eqnarray}
  for all $\delta, \delta_1, \delta_2\in\HH$ satisfying $\beta(\B(\delta_1))=\pi(\delta_2)$ and $m\in M$.
\end{defi}

Recall from Theorem \ref{thm:subadj} that a post-groupoid naturally gives to the Grossman-Larson groupoid and an action. It also naturally gives rise to a relative Rota-Baxter operator.

\begin{thm}\label{thm:post-RB}
  Let  $(\G\stackrel{\pi}{\longrightarrow}M,\cdot,\Phi,\rhd)$  be a post-groupoid. Then the identity map $\Id:\G\to\G$ is a relative Rota-Baxter operator on the Grossman-Larson groupoid $(\xymatrix{ \G \ar@<0.5ex>[r]^{\alpha=\pi} \ar[r]_{\beta=\Phi} & M},\star,\iota,\inv_\rhd)$ with respect to the action $\rhd.$
\end{thm}

\begin{proof}
  The relative Rota-Baxter relation \eqref{eq:rRB} for $\Id$ follows directly form the definition of the Grossman-Larson groupoid multiplication $\star$.
\end{proof}

Let $\B:\HH\to\G$  be a relative Rota-Baxter operator on a groupoid $\G$ with respect to an action $\rightharpoonup$ on a group bundle $\HH$. On $\HH$, define $\Phi_\B:\HH\to M$ by
\begin{equation}\label{eq:RB-Phi}
  \Phi_\B(\delta)=\beta(\B(\delta)),\quad\forall \delta\in \HH.
\end{equation}
By  \eqref{eq:unit1}, we have
  $$
  \Phi_\B(\iota^\HH_m)=\beta(\B(\iota^\HH_m))=\beta(\iota^\G_m)=m.
  $$
Define $\rhd_\B:\HH_{\Phi_\B}\times_\pi \HH\to \HH$ by
\begin{equation}\label{eq:RB-rhd}
  \delta_1\rhd_\B \delta_2=\B(\delta_1)\rightharpoonup\delta_2.
\end{equation}
It is obvious that $L^{\rhd_\B}_\delta=\B(\delta)\rightharpoonup\cdot$ is invertible. Since $\alpha (\B(\delta))=\pi(\delta)$ for all $\delta\in \HH$, we have
 $$
\pi(\delta_1\rhd_\B \delta_2)=\pi(\B(\delta_1)\rightharpoonup\delta_2)=\alpha (\B(\delta_1))=\pi(\delta_1).
 $$

\begin{thm}\label{thm:RB-post}
Let $\rightharpoonup:\G_{\beta}\times_{\pi}\HH\to \HH$ be an action of a groupoid $\xymatrix{ (\G \ar@<0.5ex>[r]^{\alpha} \ar[r]_{\beta} & M},\cdot_\G,\iota^\G,\inv_\G)$ on a group bundle $(\HH\stackrel{\pi}{\longrightarrow}M,\cdot_\HH)$.   Let $\B:\HH\to\G$  be a relative Rota-Baxter operator. Then $(\HH\stackrel{\pi}{\longrightarrow}M,\cdot_\HH,\Phi_\B,\rhd_\B)$ is a post-groupoid, where $\Phi_\B$ and $\rhd_\B$ are given by \eqref{eq:RB-Phi} and \eqref{eq:RB-rhd} respectively.
\end{thm}

\begin{proof}
  First by \eqref{eq:rRB}, for all $\gamma,\delta\in\HH$ satisfying $\beta(\B(\gamma))=\pi(\delta)$, we have
  \begin{eqnarray*}
    \Phi_\B(\gamma\cdot_\HH(\gamma\rhd_\B\delta))=\beta(\B(\gamma\cdot_\HH(\gamma\rhd_\B\delta)))=\beta(\B(\gamma)\cdot_\G\B(\delta))=\beta(\B(\delta))=\Phi_\B(\delta),
  \end{eqnarray*}
 which implies that Axiom (i)   in Definition \ref{defi:post-groupoid} holds.

 Axiom (ii) in Definition \ref{defi:action} implies that   Axiom (ii)   in Definition \ref{defi:post-groupoid} holds.

 Finally, for all $\gamma_1,\gamma_2,\gamma_3\in \HH$ satisfying $\beta(\B(\gamma_1))=\pi(\gamma_2)$ and $\beta(\B(\gamma_2))=\pi(\gamma_3)$, we have
 \begin{eqnarray*}
   \gamma_1\rhd_\B(\gamma_2\rhd_\B\gamma_3)&=&\B(\gamma_1)\rightharpoonup( \B(\gamma_2)\rightharpoonup\gamma_3)=(\B(\gamma_1)\cdot_\G\B(\gamma_2))\rightharpoonup \gamma_3\\
   &=& (\B(\gamma_1\cdot_\HH(\gamma_1\rhd_\B\gamma_2))\rightharpoonup\gamma_3=(\gamma_1\cdot_\HH(\gamma_1\rhd_\B\gamma_2))\rhd_\B\gamma_3.
 \end{eqnarray*}
  Therefore, $(\HH\stackrel{\pi}{\longrightarrow}M,\cdot_\HH,\Phi_\B,\rhd_\B)$ is a post-groupoid.
\end{proof}

The following result is a generalization of \cite[Proposition 4.28]{GLS}.

\begin{cor}\label{cor:descendent}
  Let $\B:\HH\to\G$  be a  relative Rota-Baxter operator on the groupoid $\G$ with respect to an action $\rightharpoonup$ on a group bundle $\HH$. Then  $\xymatrix{(\HH \ar@<0.5ex>[r]^{\alpha=\pi} \ar[r]_{\beta_\B} & M},\star_\B,\iota^\HH,\inv_\B)$ is a groupoid, called the {\bf descendent groupoid} of the relative Rota-Baxter operator $\B,$ where the target map $\beta_\B$, the groupoid multiplication $\star_\B$ and the inverse map $\inv_\B$ are given by
  \begin{eqnarray}
  \label{des1} \beta_\B&=&\beta\circ\B,\\
   \label{des2} \delta_1\star_\B\delta_2&=&\delta_1\cdot_\HH ( \B(\delta_1)\rightharpoonup \delta_2),\tforall \delta_1,\delta_2\in \HH ~\mbox{satisfying}~\B(\delta_1)=\pi(\delta_2),\\
  \label{des3}  \inv_\B(\delta)&=&(  \B(\delta)\rightharpoonup \cdot)^{-1}\inv_\HH(\delta).
  \end{eqnarray}

  Moreover, $\B$ is  a homomorphism from the descendent groupoid $\xymatrix{(\HH \ar@<0.5ex>[r]^{\alpha=\pi} \ar[r]_{\beta_\B} & M},\star_\B,\iota^\HH,\inv_\B)$ to the groupoid $\xymatrix{ (\G \ar@<0.5ex>[r]^{\alpha} \ar[r]_{\beta} & M},\cdot_\G,\iota^\G,\inv_\G)$.
\end{cor}

\begin{proof}
  It follows from Theorem \ref{thm:RB-post} and Theorem \ref{thm:subadj} directly.
\end{proof}

We will denote the descendent groupoid $\xymatrix{(\HH \ar@<0.5ex>[r]^{\alpha=\pi} \ar[r]_{\beta_\B} & M},\star_\B,\iota^\HH,\inv_\B)$ by $\HH_\B$.

The following lemma gives the relation between the action $\rightharpoonup$ and the descendent groupoid multiplication $\star_\B$, which plays important role in our study of matched pairs of groupoids.

\begin{lem}\label{lem:f1} For all $\gamma\in\G$ and $\delta_1,\delta_2\in\HH$ satisfying $\beta(\gamma)=\pi(\delta_1)$ and $\beta(\B(\delta_1))=\pi( \delta_2)$, we have
$$
 \gamma\rightharpoonup (\delta_1\star_\B \delta_2)=  (\gamma\rightharpoonup\delta_1)\cdot_\HH  ((\gamma\cdot_\G\B(\delta_1))\rightharpoonup \delta_2).
$$
\end{lem}
\begin{proof}
By straightforward computations, we have
    \begin{eqnarray*}
    \gamma\rightharpoonup(\delta_1\star_\B \delta_2)&=&  \gamma\rightharpoonup (\delta_1\cdot_\HH ( \B(\delta_1)\rightharpoonup\delta_2))\\
    &=&  (\gamma\rightharpoonup \delta_1)\cdot_\HH  (\gamma\rightharpoonup   ( \B(\delta_1)\rightharpoonup\delta_2))\\
    &=&   (\gamma\rightharpoonup\delta_1)\cdot_\HH  ((\gamma\cdot_\G\B(\delta_1))\rightharpoonup \delta_2),
  \end{eqnarray*}
  which finishes the proof.
\end{proof}

\begin{defi}\cite{Mackenzie}
  A {\bf matched pair of groupoids} is a pair of groupoids $(\G,\KK)$ over the same base $M$ endowed with a left action $\rightharpoonup:\G_{\beta_\G}\times_{\alpha_\KK} \KK\to \KK$ of $\G$ on $\alpha_\KK:\KK\to M$ and a right action $\leftharpoonup:\G_{\beta_\G}\times_{\alpha_\KK} \KK\to \G$ of $\KK$ on $\beta_\G:\G\to M$ satisfying
  \begin{eqnarray}
    \beta_\KK(\gamma\rightharpoonup \delta)&=&\alpha_\G(\gamma\leftharpoonup \delta),\\
    \gamma\rightharpoonup (\delta_1\cdot_\KK\delta_2)&=&(\gamma\rightharpoonup\delta_1)\cdot_\KK( (\gamma\leftharpoonup\delta_1)\rightharpoonup \delta_2),\\
    (\gamma_1\cdot_\G\gamma_2)\leftharpoonup \delta&=& (\gamma_1\leftharpoonup (\gamma_2\rightharpoonup\delta))\cdot_\G(\gamma_2 \leftharpoonup \delta),
  \end{eqnarray}
  for all $\gamma,\gamma_1,\gamma_2\in\G$ and $\delta,\delta_1,\delta_2\in\KK$ such that the compositions are possible.
\end{defi}
We usually denote a matched pair of groupoids by $(\G,\KK,\rightharpoonup,\leftharpoonup)$.

Matched pairs of groupoids can be characterized by their doubles. More precisely, $\KK_{\beta_\KK}\times_{\alpha_\G}\G$ is a groupoid over $M$, where the sauce map $\alpha$ and the target map $\beta$ are given by
$
\alpha(\delta,\gamma)=\alpha_\KK(\delta),~\beta(\delta,\gamma)=\beta_\G(\gamma),
$
and the multiplication $\cdot$ is given by
$$
(\delta_1,\gamma_1)\cdot (\delta_2,\gamma_2)=(\delta_1\cdot_\KK(\gamma_1\rightharpoonup  \delta_2),(\gamma_1\leftharpoonup  \delta_2)\cdot_\G\gamma_2).
$$

Another important feature of relative Rota-Baxter operators is that they give rise to matched pairs of groupoids.
\begin{thm}\label{thm:RB-mp}
 Let $\B:\HH\to\G$  be a  relative Rota-Baxter operator on $\xymatrix{ (\G \ar@<0.5ex>[r]^{\alpha_\G} \ar[r]_{\beta_\G} & M},\cdot_\G,\iota^\G,\inv_\G)$    with respect to an action $\rightharpoonup$ on a   group bundle $(\HH\stackrel{\pi}{\longrightarrow}M,\cdot_\HH)$. Then  $(\G,\HH_\B,\rightharpoonup,\leftharpoonup)$   form a matched pair of groupoids, where the right action $\leftharpoonup:\G_{\beta_\G}\times_{\pi}\HH\to \G$ of $\HH_\B$ on $\beta_\G:\G\to M$ is given by
 \begin{equation}\label{ra}
 \gamma\leftharpoonup \delta =\inv_\G(\B(\gamma\rightharpoonup \delta))\cdot_\G \gamma \cdot_\G\B(\delta), \tforall \delta\in \HH,~\gamma\in \G ~\mbox{satisfying}~\pi(\delta)=\beta_\G(\gamma).
\end{equation}
\end{thm}
\begin{proof} First we show that $\leftharpoonup:\G_{\beta_\G}\times_{\pi}\HH\to \G$ defined above is a right action of $\HH_\B$ on $\beta_\G:\G\to M$. Obviously, we have $\gamma\leftharpoonup \iota_{\beta_\G(\gamma)}=\gamma$ and
$$
\beta_\G( \gamma\leftharpoonup \delta) =\beta_\G(\inv^\G(\B(\gamma\rightharpoonup \delta))\cdot_\G \gamma \cdot_\G\B(\delta)) =\beta_\G( \B(\delta))=\beta_\B(\delta).
 $$
 By Lemma \ref{lem:f1} and \eqref{eq:rRB}, we have
 \begin{eqnarray*}
  && (\gamma\leftharpoonup \delta_1)\leftharpoonup \delta_2\\ &=&\Big(\inv^\G(\B(\gamma\rightharpoonup \delta_1))\cdot_\G \gamma \cdot_\G\B(\delta_1)\Big)\leftharpoonup \delta_2\\
   &=& \inv^\G\Big( \B((\inv^\G(\B(\gamma\rightharpoonup \delta_1))\cdot_\G \gamma \cdot_\G\B(\delta_1))\rightharpoonup\delta_2)\Big)\cdot_\G\Big(\inv^\G(\B(\gamma\rightharpoonup \delta_1))\cdot_\G \gamma \cdot_\G\B(\delta_1)\Big)\cdot_\G \B(\delta_2)\\
   &=&\inv^\G\Big( \B(\gamma\rightharpoonup \delta_1)\cdot_\G\B( (\inv^\G(\B(\gamma\rightharpoonup \delta_1))\cdot_\G \gamma \cdot_\G\B(\delta_1) )\rightharpoonup\delta_2)\Big)\cdot_\G\Big(\gamma \cdot_\G\B(\delta_1)\Big)\cdot_\G \B(\delta_2)\\
   &=&\inv^\G\Big( \B\big((\gamma\rightharpoonup \delta_1)\cdot_\HH(\B(\gamma\rightharpoonup \delta_1)\rightharpoonup( (\inv^\G(\B(\gamma\rightharpoonup \delta_1))\cdot_\G \gamma \cdot_\G\B(\delta_1) )\rightharpoonup\delta_2))\big)\Big)\cdot_\G\gamma \cdot_\G\B(\delta_1)\cdot_\G \B(\delta_2)\\
   &=&\inv^\G\Big( \B\big((\gamma\rightharpoonup \delta_1)\cdot_\HH(( \gamma \cdot_\G\B(\delta_1) )\rightharpoonup\delta_2)\big)\Big)\cdot_\G\gamma \cdot_\G\B(\delta_1\star_\B \delta_2)\\
   &=&\inv^\G\Big( \B\big(\gamma\rightharpoonup (\delta_1\star_\B \delta_2)\big)\Big)\cdot_\G\gamma \cdot_\G\B(\delta_1\star_\B \delta_2)\\
   &=&\gamma\leftharpoonup (\delta_1\star_\B\delta_2),
 \end{eqnarray*}
which implies that $\leftharpoonup$ is a right action.

By \eqref{ra}, we have
  \begin{eqnarray}\label{eq33}
  \alpha_\G(\gamma\leftharpoonup\delta)=\alpha_\G(\inv^\G(\B(\gamma\rightharpoonup \delta))\cdot_\G \gamma \cdot_\G\B(\delta))
=\beta_\G (\B(\gamma\rightharpoonup\delta))
= \beta_\B (\gamma\rightharpoonup\delta).
  \end{eqnarray}
 By Lemma \ref{lem:f1}, \eqref{des2}, \eqref{ra}, we have
  \begin{eqnarray}
  \label{eq11}  (\gamma\rightharpoonup\delta_1)\star_\B ((\gamma\leftharpoonup\delta_1)\rightharpoonup\delta_2)
   &=&  (\gamma\rightharpoonup\delta_1)\cdot_\HH \Big(\B(\gamma\rightharpoonup   \delta_1)\rightharpoonup  ((\gamma\leftharpoonup\delta_1)\rightharpoonup\delta_2)\Big)\\
   \nonumber &=&(\gamma\rightharpoonup\delta_1)\cdot_\HH \Big((\B(\gamma\rightharpoonup   \delta_1)\cdot_\G  (\gamma\leftharpoonup\delta_1))\rightharpoonup\delta_2\Big)\\
    \nonumber&=&(\gamma\rightharpoonup\delta_1)\cdot_\HH \Big((\gamma\cdot_\G\B(\delta_1) )\rightharpoonup\delta_2\Big)\\
    \nonumber&=& \gamma\rightharpoonup (\delta_1\star_\B \delta_2).
  \end{eqnarray}

 Apply \eqref{ra} directly, we have
  \begin{eqnarray}
  \label{eq22}  &&(\gamma_1\leftharpoonup (\gamma_2\rightharpoonup\delta))\cdot_\G(\gamma_2 \leftharpoonup \delta)\\
  \nonumber&=&\inv^\G(\B(\gamma_1\rightharpoonup (\gamma_2\rightharpoonup\delta)))\cdot_\G \gamma_1 \cdot_\G\B(\gamma_2\rightharpoonup\delta)\cdot_\G\inv^\G(\B(\gamma_2\rightharpoonup \delta))\cdot_\G \gamma_2 \cdot_\G\B(\delta)\\
   \nonumber &=&\inv^\G(\B((\gamma_1\cdot_\G\gamma_2)\rightharpoonup \delta))\cdot_\G \gamma_1\cdot_\G\gamma_2\cdot_\G\B(\delta)\\
    \nonumber&=&(\gamma_1\cdot_\G\gamma_2)\leftharpoonup\delta.
  \end{eqnarray}

  By \eqref{eq33}, \eqref{eq11} and \eqref{eq22}, we deduce that $(\G,\HH_\B,\rightharpoonup,\leftharpoonup)$ is a matched pair of groupoids.
\end{proof}

\section{Post-groupoids and quiver-theoretical solutions of the Yang-Baxter equation}\label{sec:4}

In this section, we show that post-groupoids, as well as relative Rota-Baxter operators, can be used to construct  quiver-theoretical solutions of the Yang-Baxter equation.
\begin{defi}
Let $M$ be a non-empty set. A {\bf quiver} over $M$ is a set $\A$ equipped with two maps $\alpha,\beta:\A\lon M$, called the source map and the target map, respectively. We will denote a quiver by $\xymatrix{ \A \ar@<0.5ex>[r]^{\alpha} \ar[r]_{\beta} & M}$ or simply by $\A$.
\end{defi}
\begin{rmk}
Let $\xymatrix{ (\G \ar@<0.5ex>[r]^{\alpha} \ar[r]_{\beta} & M},\cdot,\iota,\inv)$ be a groupoid. Forgetting the multiplication, the inverse map and the identity map, we obtain a quiver $\xymatrix{ \G \ar@<0.5ex>[r]^{\alpha} \ar[r]_{\beta} & M}$.
\end{rmk}
\begin{defi} A {\bf homomorphism} of quivers  from $\xymatrix{ \A \ar@<0.5ex>[r]^{\alpha_\A} \ar[r]_{\beta_\A} & M}$ to $\xymatrix{ \BB \ar@<0.5ex>[r]^{\alpha_\BB} \ar[r]_{\beta_\BB} & M}$ is a map $f:\A\to \BB$ satisfying
$
\alpha_\A=\alpha_\BB\circ f, ~
\beta_\A=\beta_\BB\circ f.
$
 In particular, if $f$ is invertible,  then $f$ is called an  {\bf isomorphism}.
\end{defi}

Let $M$ be a non-empty set. It is obvious that quivers over $M$ and homomorphisms between quivers  form a category $\Quiv(M)$. For quivers  $\xymatrix{ \A \ar@<0.5ex>[r]^{\alpha_\A} \ar[r]_{\beta_\A} & M}$ and $\xymatrix{\BB \ar@<0.5ex>[r]^{\alpha_\BB} \ar[r]_{\beta_\BB} & M}$, then $\A_{\beta_\A}\times_{\alpha_\BB}\BB$ is a quiver over $M$ with
 $$\alpha(a,b)=\alpha_\A(a),\quad\beta(a,b)=\beta_\BB(b),\quad \forall (a,b)\in \A_{\beta_\A}\times_{\alpha_\BB}\BB.$$ Thus the category $\Quiv(M)$ is a monoidal category, with   $\one=\xymatrix{ M \ar@<0.5ex>[r]^{\Id} \ar[r]_{\Id} & M}$.

Matched pairs of groupoids can be characterized by morphisms of quivers. Now we give an equivalent definition of matched pairs of groupoids.
\begin{defi}\mcite{Mackenzie}
A {\bf matched pair of groupoids} is a triple $(\G,\KK,\sigma)$, where $\G$ and $\KK$ are groupoids over the same base $M$ and
$$\sigma:\G_{\beta_\G}\times_{\alpha_\KK} \KK\lon \KK_{\beta_\KK}\times_{\alpha_\G} \G,\,\,\,\,(\gamma,\delta)\mapsto(\gamma\rightharpoonup \delta,\gamma\leftharpoonup \delta),\,\,\,\forall(\gamma,\delta)\in\G_{\beta_\G}\times_{\alpha_\KK} \KK$$
is a homomorphism of quivers satisfying the following conditions:
\begin{eqnarray}
\mlabel{MG-1}\iota_{\alpha_\KK(\delta)}\rightharpoonup \delta&=&\delta,\\
\mlabel{MG-2}\gamma_1\rightharpoonup(\gamma_2\rightharpoonup \delta)&=&(\gamma_1\gamma_2)\rightharpoonup \delta,\\
\mlabel{MG-3}(\gamma_1\gamma_2)\leftharpoonup \delta&=&\big(\gamma_1\leftharpoonup(\gamma_2\rightharpoonup \delta)\big)(\gamma_2\leftharpoonup \delta),\\
\mlabel{MG-4}\gamma\leftharpoonup \iota_{\beta_\G(\gamma)}&=& \gamma,\\
\mlabel{MG-5}(\gamma\leftharpoonup \delta_1)\leftharpoonup \delta_2&=&\gamma\leftharpoonup(\delta_1\delta_2),\\
\mlabel{MG-6}\gamma\rightharpoonup(\delta_1\delta_2)&=&(\gamma\rightharpoonup \delta_1)\big((\gamma\leftharpoonup \delta_1)\rightharpoonup \delta_2\big),
\end{eqnarray}
for all $\gamma,\gamma_1,\gamma_2\in\G$ and $\delta,\delta_1,\delta_2\in\KK$ such that the compositions are possible.
\end{defi}
The following concept was introduced by Andruskiewitsch in \mcite{Andruskiewitsch-1}.

\begin{defi}\mcite{Andruskiewitsch-1}
Let $\xymatrix{\A \ar@<0.5ex>[r]^{\alpha} \ar[r]_{\beta} & M}$ be a quiver. A quiver-theoretical solution of the {\bf Yang-Baxter equation} on the quiver $\A$ is an isomorphism $R$ of  quivers from  $\A_{\beta}\times_{\alpha}\A$ to $\A_{\beta}\times_{\alpha}\A$ satisfying:
\begin{eqnarray}
R_{12}R_{23}R_{12}=R_{23}R_{12}R_{23}:\A_{\beta}\times_{\alpha}\A_{\beta}\times_{\alpha}\A\lon \A_{\beta}\times_{\alpha}\A_{\beta}\times_{\alpha}\A,
\end{eqnarray}
where $R_{12}=R\times\Id_\A,~R_{23}=\Id_\A\times R$.

A quiver $\A$ with a quiver-theoretical solution of
the Yang-Baxter equation on $\A$ is called a {\bf braided quiver} and
is denoted   by $(\A,R)$. Moreover, we denote by
$R(x,y)=(x\rightharpoonup y,x\leftharpoonup y)$
for all $(x,y)\in \A_{\beta}\times_{\alpha}\A$.   $R$  is called {\bf non-degenerate} if for all $x\in \A$, the maps $x\rightharpoonup\cdot:\alpha^{-1}(\beta(x))\to\alpha^{-1}(\alpha(x))$ and $\cdot\leftharpoonup x:\beta^{-1}(\alpha(x))\to\beta^{-1}(\beta(x))$ are bijective.
\end{defi}

\begin{defi}
  Let  $(\A,R)$ and $(\A',R')$ be braided quivers. A {\bf homomorphism of braided quivers}   from $(\A,R)$ to $(\A',R')$ is a homomorphism of quivers $f:\A\to \A'$ such that
$
(f\times f)\circ R=R'\circ(f\times f).
$
\end{defi}

Non-degenerate braided quivers and homomorphisms between them form a category $\BS$.

\begin{rmk}
A quiver-theoretical solution of the Yang-Baxter equation on a quiver over $M$ is a  solution of the Yang-Baxter equation in the monoidal category $\Quiv(M)$.
\end{rmk}

\begin{defi} \mcite{Andruskiewitsch-1,MM}
A {\bf braided groupoid} over $M$ is a pair $(\G,\sigma)$, where $\G$ is a groupoid and $\sigma:\G_\beta\times_\alpha \G\lon \G_\beta\times_\alpha \G$ is a homomorphism of quivers such that
\begin{enumerate}
\item[\rm(i)] the triple $(\G,\G,\sigma)$ is a matched pair of groupoids,
\item[\rm(ii)] $(\gamma\rightharpoonup \delta)(\gamma \leftharpoonup \delta)=\gamma\delta, \tforall (\gamma,\delta)\in \G_\beta\times_\alpha \G$.
\end{enumerate}

For braided groupoids $(\G,\sigma)$ and $(\G',\sigma')$, a homomorphism of braided groupoids from $(\G,\sigma)$ to $(\G',\sigma')$ is a homomorphism of groupoids $\Psi:\G\to \G'$ such that
$
(\Psi\times \Psi)\sigma=\sigma'(\Psi\times \Psi).
$
\end{defi}

Braided groupoids and homomorphisms between them form a category $\BG$.

\begin{thm}{\rm \mcite{Andruskiewitsch-1}}\mlabel{bgybe}
Let $(\G,\sigma)$ be a braided groupoid. Then $\sigma:\G_\beta\times_\alpha \G\lon \G_\beta\times_\alpha \G$ is a non-degenerate  quiver-theoretical solution of the  Yang-Baxter equation on the quiver $\xymatrix{\G \ar@<0.5ex>[r]^{\alpha} \ar[r]_{\beta} & M}$.
\end{thm}

Let $(\G\stackrel{\pi}{\longrightarrow}M,\Phi,\rhd)$ be a post-groupoid. Define $R_\G:\G_\Phi\times_\pi \G\lon \G_\Phi\times_\pi \G$ by
\begin{eqnarray}
  R_\G(\gamma,\delta)=\big(\gamma\rhd\delta,\inv_\rhd(\gamma\rhd\delta)\star \gamma\star \delta\big), \quad \forall (\gamma,\delta)\in \G_\Phi\times_\pi \G,
\end{eqnarray}
  where $\star$ is the Grossman-Larson groupoid structure given in Theorem \mref{thm:subadj}.

\begin{thm}\mlabel{pgybe}
Let $(\G\stackrel{\pi}{\longrightarrow}M,\cdot,\Phi,\rhd)$ be a post-groupoid. Then the Grossman-Larson groupoid $(\xymatrix{\G \ar@<0.5ex>[r]^{\alpha=\pi} \ar[r]_{\beta=\Phi} & M},\star,\iota,\inv_\rhd)$ together with  $R_\G$ defined above form a braided groupoid, and $R_\G$ is a non-degenerate quiver-theoretical solution of the Yang-Baxter equation on the quiver  $\xymatrix{\G \ar@<0.5ex>[r]^{\alpha=\pi} \ar[r]_{\beta=\Phi} & M}.$

Moreover, let $\Psi$ be a  homomorphism of post-groupoids from  $(\G,\cdot_\G,\Phi_\G,\rhd_\G)$ to $(\HH,\cdot_\HH,\Phi_\HH,\rhd_\HH)$. Then $\Psi$ is a homomorphism of the corresponding braided groupoids.
\end{thm}

\begin{proof}
First by Theorem \ref{thm:post-RB},   the identity map $\Id:\G\to\G$ is a relative Rota-Baxter operator on the Grossman-Larson groupoid $(\xymatrix{\G \ar@<0.5ex>[r]^{\alpha=\pi} \ar[r]_{\beta=\Phi} & M},\star,\iota,\inv_\rhd)$ with respect to the action $\rhd.$ Then by Theorem \ref{thm:RB-mp}, the Grossman-Larson groupoid $(\xymatrix{\G \ar@<0.5ex>[r]^{\alpha=\pi} \ar[r]_{\beta=\Phi} & M},\star,\iota,\inv_\rhd)$ and itself form a matched pair, where the left action $ \rightharpoonup$ is exactly $\rhd$, and the right action $\leftharpoonup$ is given by \eqref{ra}. Since $\B=\Id$, so the right action $\leftharpoonup$ is given explicitly by
$
\gamma\leftharpoonup \delta=\inv_\rhd(\gamma\rhd\delta)\star \gamma\star \delta.
$
Moreover, for all $(\gamma,\delta)\in \G_\beta\times_\alpha \G$, we have
\begin{eqnarray*}
(\gamma\rightharpoonup \delta)\star(\gamma \leftharpoonup \delta)&=&(\gamma\rhd\delta)\star \inv_\rhd(\gamma\rhd\delta)\star \gamma\star \delta= \gamma\star \delta,
\end{eqnarray*}
which implies that the Grossman-Larson groupoid $(\xymatrix{ \G \ar@<0.5ex>[r]^{\alpha=\pi} \ar[r]_{\beta=\Phi} & M},\star,\iota,\inv_\rhd)$ together with  $R_\G$ defined above form a braided groupoid. By Theorem \ref{bgybe}, $R_\G$ is a non-degenerate quiver-theoretical solution of the Yang-Baxter equation on the quiver  $\xymatrix{\G \ar@<0.5ex>[r]^{\alpha=\pi} \ar[r]_{\beta=\Phi} & M}.$

Let $\Psi$ be a  homomorphism of post-groupoids from  $(\G,\cdot_\G,\Phi_\G,\rhd_\G)$ to $(\HH,\cdot_\HH,\Phi_\HH,\rhd_\HH)$. By Proposition \ref{pro:subadj-mor},   $\Psi$ is a groupoid homomorphism between the corresponding Grossman-Larson groupoids. Furthermore, we have
\begin{eqnarray*}
  \Psi\times \Psi (R_\G(\gamma,\delta))&=&\Psi\times \Psi\big(\gamma\rhd_\G\delta,\inv_{\rhd_\G}(\gamma\rhd_\G\delta)\star_\G \gamma\star_\G \delta\big)\\
  &=&\big(\Psi(\gamma\rhd_\G\delta),\Psi(\inv_{\rhd_\G}(\gamma\rhd_\G\delta)\star_\G \gamma\star_\G \delta)\big)\\
    &=&\big(\Psi(\gamma)\rhd_\HH\Psi(\delta),\inv_{\rhd_\HH}(\Psi(\gamma)\rhd_\HH\Psi(\delta))\star_\HH \Psi(\gamma)\star_\HH \Psi(\delta)\big)\\
    &=&R_\HH(\Psi(\gamma),\Psi(\delta)),
\end{eqnarray*}
which implies that $\Psi$ is a homomorphism between braided groupoids.
\end{proof}

Since a relative Rota-Baxter operator induces a post-groupoid, we have the following corollary.
\begin{cor}
Let $\B:\HH\to \G$ be a  relative Rota-Baxter operator on a groupoid $\G$ with respect to an action $\Phi$ on a group bundle $\HH$. Then $R_\B:\HH_{\beta\circ\B}\times_\pi \HH\to \HH_{\beta\circ\B}\times_\pi \HH$ defined by
\begin{eqnarray}
R_\B(h,k)=\Big(\Phi(\B(h))(k),\inv_\B(\Phi(\B(h))(k)) \star_\B h\star_\B k\Big),
\end{eqnarray}
is a non-degenerate solution of the  Yang-Baxter equation on the quiver $\xymatrix{\HH \ar@<0.5ex>[r]^{\pi} \ar[r]_{\beta\circ\B} & M}$, where $\star_\B$ is the descendent groupoid structure given by Corollary \ref{cor:descendent}.
\end{cor}

\begin{ex}
  Consider the post-groupoid $(M\times G\stackrel{\pi}{\longrightarrow}M,\cdot,\Phi,\rhd)$ given in Proposition  \ref{ex:post-g} associated to an action $\Phi:M\times G\to M$. In this case, $$(M\times G)_\Phi\times_\pi (M\times G)=\{((m,g),(n,h))\in (M\times G)\times (M\times G) ~\mbox{such that}~ n=\Phi(m,g)\}.$$
 Then $R_{M\times G}:(M\times G)_\Phi\times_\pi (M\times G)\lon (M\times G)_\Phi\times_\pi (M\times G)$ defined by
\begin{eqnarray}
  R_{M\times G}((m,g),(n,h))=\big((m,h),(\Phi(m,h),h^{-1}gh)\big),
\end{eqnarray}
for all $((m,g),(n,h))\in (M\times G)_\Phi\times_\pi (M\times G)$, is a non-degenerate quiver-theoretical solution of the Yang-Baxter equation on the quiver $\xymatrix{M\times G \ar@<0.5ex>[r]^{~~\pi} \ar[r]_{~~\Phi} & M}$.
\end{ex}

Let $\big((\xymatrix{ \G \ar@<0.5ex>[r]^{\alpha} \ar[r]_{\beta} & M},\star,\iota,\inv),\sigma\big)$ be a braided groupoid. We write
$\sigma(\gamma,\delta)=(\gamma\rightharpoonup \delta,\gamma\leftharpoonup \delta)$ for all $(\gamma,\delta)\in
\G_\beta\times_\alpha\G$. Define $\rhd:\G_\beta\times_\alpha\G\lon \G$ by
\begin{eqnarray}\label{eq:c1}
\gamma\rhd\delta=\gamma\rightharpoonup \delta\tforall (\gamma,\delta)\in
\G_\beta\times_\alpha\G.
\end{eqnarray}
Define  $\cdot:\G_\alpha\times_\alpha \G\lon
\G$ by
\begin{eqnarray}\label{eq:c2}
\gamma_1\cdot \gamma_2=\gamma_1\star(\inv(\gamma_1)\rightharpoonup \gamma_2) \tforall (\gamma_1,\gamma_2)\in \G_\alpha\times_\alpha \G.
\end{eqnarray}

\begin{thm}\mlabel{ybepg}
Let $\big((\xymatrix{ \G \ar@<0.5ex>[r]^{\alpha} \ar[r]_{\beta} & M},\star,\iota,\inv),\sigma\big)$ be a braided groupoid. Then $(\G\stackrel{\alpha}{\longrightarrow}M,\cdot,\beta,\rhd)$ is a post-groupoid, where $\rhd$ and $\cdot$ are given by \eqref{eq:c1} and \eqref{eq:c2} respectively.
Moreover, let $\phi$ be a  homomorphism
  from the braided groupoid $\big((\xymatrix{ \G \ar@<0.5ex>[r]^{\alpha_\G} \ar[r]_{\beta_\G} & M},\star_\G,\iota_\G,\inv_\G),\sigma_\G\big)$ to the braided groupoid
$\big((\xymatrix{ \HH \ar@<0.5ex>[r]^{\alpha_\HH} \ar[r]_{\beta_\HH} & M},\star_\HH,\iota_\HH,\inv_\HH),\sigma_\HH\big)$. Then $\phi$ is a homomorphism of
post-groupoids from $(\G\stackrel{\alpha_\G}{\longrightarrow}M,\cdot_\G,\beta_\G,\rhd_\G)$ to $(\HH\stackrel{\alpha_\HH}{\longrightarrow}M,\cdot_\HH,\beta_\HH,\rhd_\HH)$.
\end{thm}
\begin{proof}
For $(\gamma_1,\gamma_2)\in \G_\alpha\times_\alpha \G$, we have
\begin{eqnarray*}
\beta(\inv(\gamma_1))&=&\alpha(\gamma_1)=\alpha(\gamma_2),\\
\alpha(\inv(\gamma_1)\rightharpoonup \gamma_2)&=&\alpha(\inv(\gamma_1))=\beta(\gamma_1),\\
\alpha(\gamma_1\star(\inv(\gamma_1)\rightharpoonup \gamma_2))&=&\alpha(\gamma_1),
\end{eqnarray*}
thus the definition of $\gamma_1\cdot \gamma_2$ is well-defined. For all $\gamma\in \G$, we have
\begin{eqnarray*}
\gamma\cdot \iota_{\alpha(\gamma)}&=&\gamma\star(\inv(\gamma)\rightharpoonup \iota_{\alpha(\gamma)})\stackrel{\meqref{MG-4},\meqref{MG-6}}{=}\gamma,\\
\gamma\cdot (\gamma\rightharpoonup \inv(\gamma))&=&\gamma\star \big(\inv(\gamma)\rightharpoonup (\gamma\rightharpoonup \inv(\gamma))\big)\stackrel{\meqref{MG-1},\meqref{MG-2}}{=}\iota_{\alpha(\gamma)}.
\end{eqnarray*}
Thus, $\iota$ is  a right unit of the multiplication $\cdot$ and
$\gamma\rightharpoonup \inv(\gamma)$ is  a  right inverse of $\gamma\in \G.$ For all $\gamma_1,\gamma_2,\gamma_3\in \G$  satisfying $\beta(\gamma_1)=\alpha(\gamma_2)=\alpha(\gamma_3)$, we have
\begin{eqnarray*}
\gamma_1\rhd (\gamma_2\cdot\gamma_3)&=&\gamma_1\rightharpoonup \big(\gamma_2\star (\inv(\gamma_2)\rightharpoonup \gamma_3)\big) \stackrel{\meqref{MG-6}}{=}(\gamma_1\rightharpoonup \gamma_2)\star \Big((\gamma_1\leftharpoonup \gamma_2)\rightharpoonup(\inv(\gamma_2)\rightharpoonup \gamma_3)\Big),\\
(\gamma_1\rhd \gamma_2)\cdot (\gamma_1\rhd \gamma_3)&=&(\gamma_1\rightharpoonup \gamma_2)\cdot (\gamma_1\rightharpoonup \gamma_3)=(\gamma_1\rightharpoonup \gamma_2)\star\Big(\inv(\gamma_1\rightharpoonup \gamma_2)\rightharpoonup (\gamma_1\rightharpoonup \gamma_3)\Big).
\end{eqnarray*}
By $(\gamma_1\rightharpoonup \gamma_2)\star(\gamma_1\leftharpoonup \gamma_2)=\gamma_1\star\gamma_2$, we have
\begin{eqnarray*}
(\gamma_1\rightharpoonup \gamma_2)\rightharpoonup\Big((\gamma_1\leftharpoonup \gamma_2)\rightharpoonup(\inv(\gamma_2)\rightharpoonup \gamma_3)\Big)&\stackrel{\meqref{MG-2}}{=}&(\gamma_1\star\gamma_2)\rightharpoonup(\inv(\gamma_2)\rightharpoonup \gamma_3)\stackrel{\meqref{MG-2}}{=}\gamma_1\rightharpoonup \gamma_3\\
&=&(\gamma_1\rightharpoonup \gamma_2)\rightharpoonup\Big(\inv(\gamma_1\rightharpoonup \gamma_2)\rightharpoonup (\gamma_1\rightharpoonup \gamma_3)\Big).
\end{eqnarray*}
By \eqref{MG-1}, we obtain $(\gamma_1\leftharpoonup \gamma_2)\rightharpoonup(\inv(\gamma_2)\rightharpoonup \gamma_3)=\inv(\gamma_1\rightharpoonup \gamma_2)\rightharpoonup (\gamma_1\rightharpoonup \gamma_3)$.
Moreover, we have
\begin{eqnarray}\mlabel{rhd-action}
\gamma_1\rhd (\gamma_2\cdot\gamma_3)=(\gamma_1\rhd \gamma_2)\cdot (\gamma_1\rhd \gamma_3).
\end{eqnarray}
For all $\gamma_1,\gamma_2,\gamma_3\in \G$  satisfying $\alpha(\gamma_1)=\alpha(\gamma_2)=\alpha(\gamma_3)$, we have
\begin{eqnarray*}
\gamma_1\cdot(\gamma_2\cdot\gamma_3)&=&\gamma_1\star\Big(\inv(\gamma_1)\rightharpoonup(\gamma_2\cdot\gamma_3)\Big)\\
&\stackrel{\meqref{rhd-action}}{=}& \gamma_1\star\Big((\inv(\gamma_1)\rightharpoonup \gamma_2)\cdot (\inv(\gamma_1)\rightharpoonup \gamma_3)\Big)\\
&=&\gamma_1\star\Big((\inv(\gamma_1)\rightharpoonup \gamma_2)\star \big(\inv(\inv(\gamma_1)\rightharpoonup \gamma_2)\rightharpoonup(\inv(\gamma_1)\rightharpoonup \gamma_3)\big)\Big)\\
&=&\big(\gamma_1\star(\inv(\gamma_1)\rightharpoonup \gamma_2)\big) \star\Big(\inv\big(\inv(\gamma_1)\rightharpoonup \gamma_2\big)\rightharpoonup\big(\inv(\gamma_1)\rightharpoonup \gamma_3\big)\Big)\\
&\stackrel{\meqref{MG-2}}{=}&\big(\gamma_1\star(\inv(\gamma_1)\rightharpoonup \gamma_2)\big) \star\Big(\inv\big(\gamma_1\star(\inv(\gamma_1)\rightharpoonup \gamma_2)\big)\rightharpoonup \gamma_3)\Big)\\
&=&(\gamma_1\cdot \gamma_2)\cdot \gamma_3.
\end{eqnarray*}
Thus, we deduce that $(\G\stackrel{\alpha}{\longrightarrow}M,\cdot,\iota)$ is a group bundle.

Since $\iota$ is the identity map of the groupoid $(\xymatrix{ \G \ar@<0.5ex>[r]^{\alpha} \ar[r]_{\beta} & M},\star,\iota,\inv)$, we obtain that $\beta\circ \iota=\Id$. Due to the fact that $\sigma:\G_\beta\times_\alpha\G\to \G_\beta\times_\alpha\G$ is a homomorphism of quivers, we obtain that
$$
\alpha(\gamma\rhd\delta)=\alpha(\gamma),\tforall (\gamma,\delta)\in
\G_\beta\times_\alpha\G.
$$

By \eqref{MG-1} and \eqref{MG-2}, we deduce that $$L^\rhd_\gamma:\alpha^{-1}(\beta(\gamma))\to \alpha^{-1}(\alpha(\gamma)), \quad L^\rhd_\gamma \delta= \gamma\rhd \delta \tforall \delta\in \alpha^{-1}(\beta(\gamma)),$$
is invertible for any $\gamma\in\G$.  For all $(\gamma_1,\gamma_2)\in \G_\beta\times_\alpha \G$, we  have
\begin{eqnarray}\mlabel{GL}
\gamma_1\cdot(\gamma_1\rhd \gamma_2)&=&\gamma_1\star\big(\inv(\gamma_1)\rightharpoonup(\gamma_1\rightharpoonup \gamma_2)\big)\stackrel{\meqref{MG-1},\meqref{MG-2}}{=}\gamma_1\star\gamma_2.
\end{eqnarray}
Thus, we have $\beta(\gamma_1\cdot(\gamma_1\rhd \gamma_2))=\beta(\gamma_1\star\gamma_2)=\beta(\gamma_2)$. Moreover, for all $(\gamma_1,\gamma_2,\gamma_3)\in \G_\beta\times_\alpha \G_\beta\times_\alpha \G$, we  deduce that
\begin{eqnarray*}
\gamma_1\rhd (\gamma_2\rhd \gamma_3)\stackrel{\meqref{MG-2}}{=}(\gamma_1\star\gamma_2)\rhd \gamma_3\stackrel{\meqref{GL}}{=}(\gamma_1\cdot(\gamma_1\rhd \gamma_2))\rhd \gamma_3.
\end{eqnarray*}
Thus, we obtain that $(\G\stackrel{\alpha}{\longrightarrow}M,\cdot,\beta,\rhd)$ is a post-groupoid.

Let $\phi$ be a  homomorphism
 from the braided groupoid $\big((\xymatrix{ \G \ar@<0.5ex>[r]^{\alpha_\G} \ar[r]_{\beta_\G} & M},\star_\G,\iota_\G,\inv_\G),\sigma_\G\big)$ to the braided groupoid
$\big((\xymatrix{ \HH \ar@<0.5ex>[r]^{\alpha_\HH} \ar[r]_{\beta_\HH} & M},\star_\HH,\iota_\HH,\inv_\HH),\sigma_\HH\big)$.  Then $\phi:\G\lon \HH$ is a homomorphism of groupoids  such that
\begin{eqnarray*}
\phi(\gamma \rightharpoonup \delta)=\phi(\gamma)\rightharpoonup \phi(\delta),\quad
\phi(\gamma \leftharpoonup \delta)=\phi(\gamma)\leftharpoonup \phi(\delta) \tforall (\gamma,\delta)\in
\G_{\beta_{\G}}\times_{\alpha_{\G}}\G.
\end{eqnarray*}
Thus, we have $\alpha_{\G}=\alpha_{\HH}\circ \phi$ and $\beta_{\G}=\beta_{\HH}\circ \phi$.
Furthermore, for all $(\gamma,\delta)\in
\G_{\beta_{\G}}\times_{\alpha_{\G}}\G$ and $(\gamma_1,\gamma_2)\in \G_{\alpha_{\G}}\times_{\alpha_{\G}} \G$, we have
\begin{eqnarray*}
\phi(\gamma\rhd_{\G} \delta)&=&\phi(\gamma\rightharpoonup \delta)=\phi(\gamma)\rightharpoonup \phi(\delta)=\phi(\gamma)\rhd_{\HH} \phi(\delta),\\
\phi(\gamma_1\cdot_\G \gamma_2)&=&\phi\big(\gamma_1\star_\G (\inv(\gamma_1)\rightharpoonup \gamma_2)\big)=\phi(\gamma_1)\star_{\HH} \phi(\inv(\gamma_1)\rightharpoonup \gamma_2)\\
&=&\phi(\gamma_1)\star_{\HH} \big(\inv(\phi(\gamma_1))\rightharpoonup \phi(\gamma_2)\big)=\phi(\gamma_1)\cdot_{\HH} \phi(\gamma_2),
\end{eqnarray*}
which implies that  $\phi$ is a homomorphism of post-groupoids from $(\G\stackrel{\alpha_\G}{\longrightarrow}M,\cdot_\G,\beta_\G,\rhd_\G)$ to $(\HH\stackrel{\alpha_\HH}{\longrightarrow}M,\cdot_\HH,\beta_\HH,\rhd_\HH)$. The proof is finished.
\end{proof}

Moreover, we have
\begin{thm}\mlabel{pgybef}
Theorem \mref{pgybe} defines a functor $\huaB:\PG\lon
\BG$, and Theorem \mref{ybepg} defines a functor $\huaP:\BG\lon
  \PG$ that is the inverse of $\huaB$, giving an isomorphism between the categories $\PG$ and $\BG$.
\end{thm}
\begin{proof}
Let $(\G\stackrel{\pi}{\longrightarrow}M,\cdot,\Phi,\rhd)$ be a post-groupoid. By Theorems \mref{pgybe} and \mref{ybepg}, we have
$$(\huaP\huaB)(\G\stackrel{\pi}{\longrightarrow}M,\cdot,\Phi,\rhd)=(\G\stackrel{\pi}{\longrightarrow}M,\cdot,\Phi,\rhd).$$ Let $\Psi$ be a  homomorphism of post-groupoids from  $(\G,\cdot_\G,\Phi_\G,\rhd_\G)$ to $(\HH,\cdot_\HH,\Phi_\HH,\rhd_\HH)$. Theorems \mref{pgybe} and
\mref{ybepg} yield $(\huaP\huaB)(\Psi)=\Psi$. Thus, $\huaP\huaB=\Id$.
Similarly, $\huaB\huaP=\Id$. Therefore, the functor $\huaB$ is the
inverse of the functor $\huaP$.
\end{proof}

By Theorem \mref{bgybe},   a braided groupoid $(\G,\sigma)$ is a non-degenerate braiding quiver by forgetting the groupoid structure. So there is a forgetful functor
$$\huaU:\BG\lon \BS$$
from the category of braided groupoids to the category of non-degenerate braided quivers.

Let $(\A,R)$ be a non-degenerate braided quiver. The structure groupoid $\G_\A$ is defined to be the groupoid generated by the quiver $\xymatrix{ \A \ar@<0.5ex>[r]^{\alpha} \ar[r]_{\beta} & M}$  with the defining relations
\begin{eqnarray}
xy=(x\rightharpoonup y)(x\leftharpoonup y)\quad  \mbox{when} \quad R(x,y)=(x\rightharpoonup y,x\leftharpoonup y)\tforall (x,y)\in \A_{\beta}\times_{\alpha}\A.
\end{eqnarray}
The construction of the structure groupoid $\G_\A$ was given by Andruskiewitsch in \cite[Definition 2.3]{Andruskiewitsch-1}. Further, Andruskiewitsch extended the bijective map $R$ to a braided groupoid structure $\bar{R}$ on the structure groupoid $\G_\A$ \cite[Theorem 3.8]{Andruskiewitsch-1}. Moreover, Andruskiewitsch's construction of the braided groupoids is functorial, i.e. there is a structure groupoid functor $\huaF:\BS\lon\BG$   given by
\begin{eqnarray*}
\huaF(\A,R)=(\G_\A,\bar{R}) \tforall (\A,R)\in  \BS.
\end{eqnarray*}
By the universal property of the structure groupoid $\G_\A$,  the functor $\huaF$ is the left adjoint functor of the forgetful functor $\huaU$.

Moreover,  we have the functor  $\huaU\circ \huaB:\PG\lon
\BS$ from the category of post-groupoids to the category of non-degenerate braided quivers, and the functor $\huaP\circ \huaF:\BS\lon
\PG$ from the category of non-degenerate braided quivers to the category of post-groupoids.

\begin{thm}\mlabel{bspg}
With the above notations,  the functor $\huaP\circ \huaF$ is left adjoint to $\huaU\circ \huaB$.
\end{thm}
\begin{proof}
It follows directly from Theorem  \mref{pgybef} and $\huaF$ is the left adjoint functor of the forgetful functor $\huaU$.
\end{proof}
\section{Post-groupoids and skew-left bracoids}\label{sec:5}

In this section, we introduce the notion of skew-left bracoids, which are generalizations of skew-left braces. We show that there is a one-to-one correspondence between skew-left bracoids and post-groupoids.

\begin{defi}{\rm \mcite{GV,Vendramin}}
A skew-left brace  $(G,\star,\cdot)$ consists of a group $(G,\cdot)$ and a group $(G,\star)$ such that
 \begin{eqnarray}
a\star (b\cdot c)=(a\star b)\cdot a^{-1}\cdot (a\star c)\tforall a,b,c\in G.
\end{eqnarray}
Here $a^{-1}$ is the inverse of $a$ in the group $(G,\cdot)$.

\end{defi}

\begin{defi}
A {\bf skew-left bracoid} consists of a  group bundle $(\G\stackrel{\pi}{\longrightarrow}M,\cdot,\iota,\inv)$ and a   groupoid $(\xymatrix{ \G \ar@<0.5ex>[r]^{\alpha=\pi }\ar[r]_{\beta} & M},\star,\iota,\overline{\inv})$  such that
 \begin{eqnarray}\label{eq:bracoid}
\gamma\star (\delta\cdot \delta')=(\gamma\star \delta)\cdot \inv(\gamma)\cdot (\gamma\star \delta'),
\end{eqnarray}
for all $\gamma, \delta,\delta'\in \G$ satisfying $\beta(\gamma)=\pi(\delta)=\pi(\delta')$.

A {\bf bracoid} is a skew-left bracoid in which the group bundle is abelian.
\end{defi}

\begin{rmk}
Since there is a close relationship between skew-left braces and  Hopf-Galois extensions of Hopf algebras \cite{Bachiller}, we hope that skew-left bracoids will play an important role in the study of Hopf Galois extensions of Hopf algebroids \cite{Han-Schauenburg} in the future.
\end{rmk}

Note that in the middle term of the right hand side of \eqref{eq:bracoid}, $\inv(\gamma)$ is the inverse of $\gamma$ in the group bundle $(\G\stackrel{\pi}{\longrightarrow}M,\cdot,\iota,\inv)$. Therefore, $\gamma\star \delta$, $ \inv(\gamma)$ and $\gamma\star \delta'$ are all in the fiber of $\pi^{-1}(\pi(\gamma))$, so that the right hand side of \eqref{eq:bracoid} is well defined.


\begin{rmk}
Obviously, if $M$ reduces to a point, then a skew-left bracoid reduces to a skew-left brace \cite{CSV,GV,Vendramin}, and a bracoid reduces to a brace \cite{CJO,Ru,Sm18}. Recently, Martin-Lyons and Truman introduced a new algebraic object which is also called a skew-left bracoid \cite[Definition 2.1]{MT}. It is not the same as the one introduced here. 
\end{rmk}

\begin{ex}
 Let $\Phi:M\times G\to M$ be a right action of a group $G$   on a set $M$. Then the group bundle $\G=M\times G\stackrel{\pi}{\longrightarrow}M$, where $\pi$ is the projection to $M$, and the action groupoid $M\times G\rightrightarrows M$ given in Example \ref{ex:action-g} form a skew-left bracoid. In fact, for all $(m,g),~(n,h),~(n,k)\in M\times G$ satisfying $n=\Phi(m,g)$, we have
\begin{eqnarray*}
 (m,g)\star\big((n,h)\cdot(n,k)\big)&=&(m,g)\star(n,h\cdot k)=(m,g\cdot h\cdot k),\\
 \big((m,g)\star(n,h)\big)  \cdot\inv(m,g)\cdot\big((m,g)\star(n,k)\big)  &=& (m,g\cdot h)\cdot(m,g^{-1})\cdot (m,g\cdot k) \\
 &=&(m,(g\cdot h)\cdot g^{-1}\cdot (g\cdot k))=(m,g\cdot h\cdot k),
 \end{eqnarray*}
 which implies that the group bundle $\G=M\times G\stackrel{\pi}{\longrightarrow}M$  and the action groupoid $M\times G\rightrightarrows M$   form a skew-left bracoid.

\end{ex}

Note that an action naturally induces a post-groupoid, and the above example can be generalized to arbitrary post-groupoid as the following theorem shows.

\begin{thm}\label{thm:post-brace}
 Let $(\G\stackrel{\pi}{\longrightarrow}M,\cdot,\Phi,\rhd)$ be a post-groupoid. Then the  group bundle $(\G\stackrel{\pi}{\longrightarrow}M,\cdot,\iota,\inv)$ and the Grossman-Larson groupoid $(\xymatrix{ \G \ar@<0.5ex>[r]^{\alpha=\pi} \ar[r]_{\beta=\Phi} & M},\star,\iota,\inv_\rhd)$ form a skew-left bracoid.
\end{thm}

\begin{proof}
For all $\gamma, \delta,\delta'\in \G$ satisfying $\beta(\gamma)=\pi(\delta)=\pi(\delta')$, by Axiom (ii)   in Definition \ref{defi:post-groupoid}, we have
  \begin{eqnarray*}
   \gamma\star (\delta\cdot \delta')&=& \gamma\cdot (\gamma\rhd  (\delta\cdot \delta'))=\gamma\cdot\Big(  ( \gamma\rhd\delta) \cdot (\gamma\rhd\delta')  \Big)\\
   &=&\Big(\gamma\cdot  ( \gamma\rhd\delta)\Big) \cdot \inv(\gamma)\cdot\Big(\gamma \cdot (\gamma\rhd\delta')  \Big)=(\gamma\star \delta)\cdot \inv(\gamma)\cdot (\gamma\star \delta').
  \end{eqnarray*}
  Therefore, the  group bundle $(\G\stackrel{\pi}{\longrightarrow}M,\cdot,\iota,\inv)$ and the Grossman-Larson groupoid form a skew-left bracoid.
\end{proof}

\begin{cor}
  Let $\Psi:\G_{\beta}\times_{\pi}\HH\to \HH$ be an action of a groupoid $\xymatrix{ (\G \ar@<0.5ex>[r]^{\alpha} \ar[r]_{\beta} & M},\cdot_\G,\iota^\G,\inv_\G)$ on a group bundle $\HH\stackrel{\pi}{\longrightarrow}M$.   Let $\B:\HH\to\G$  be a relative Rota-Baxter operator. Then the  group bundle $\HH\stackrel{\pi}{\longrightarrow}M$ and the descendent groupoid  $\xymatrix{(\HH \ar@<0.5ex>[r]^{\alpha=\pi} \ar[r]_{\beta_\B} & M},\star_\B,\iota^\HH,\inv_\B)$ given in Corollary \ref{cor:descendent} form a skew-left bracoid.
\end{cor}
\begin{proof}
  It follows from Theorem \ref{thm:post-brace} and Theorem \ref{thm:RB-post} directly.
\end{proof}

\begin{thm}\label{thm:bracoid-post}
  Let a group bundle $(\G\stackrel{\pi}{\longrightarrow}M,\cdot,\iota,\inv)$ and a groupoid $(\xymatrix{ \G \ar@<0.5ex>[r]^{\alpha=\pi} \ar[r]_{\beta} & M},\star,\iota,\overline{\inv})$ be a skew-left bracoid. Define $\Phi=\beta:\G\to M$ and define $\rhd:\G_\beta\times_\pi\G\to \G$ by
  \begin{equation}
    \gamma\rhd\delta=\inv(\gamma)\cdot(\gamma\star\delta).
  \end{equation}
  Then $(\G\stackrel{\pi}{\longrightarrow}M,\cdot,\Phi,\rhd)$ is a post-groupoid.
\end{thm}
\begin{proof}
  First it is obvious that
  $
  \Phi(\iota_m)=\beta(\iota_m)=m,$ for all $m\in M,
  $
  and $\pi(\gamma\rhd \delta)=\pi(\inv(\gamma)\cdot(\gamma\star\delta))=\pi(\gamma)$. It is also straightforward to see that $\gamma\rhd\delta=\iota_{\pi(\gamma)}$ if and only if $\delta=\iota_{\beta(\gamma)}$, which implies that $L^\rhd_\gamma:\G_{\Phi(\gamma)}\to\G_{\pi(\gamma)}$ is invertible.

  For all $\gamma,\delta\in\G$ satisfying $\Phi(\gamma)=\pi(\delta)$, we have
  \begin{eqnarray*}
    \Phi(\gamma\cdot(\gamma\rhd\delta))&=&\beta(\gamma\star\delta)=\beta(\delta)=\Phi(\delta),
  \end{eqnarray*}
  which implies that Axiom (i)   in Definition \ref{defi:post-groupoid} holds.

  For all $\gamma, \delta,\delta'\in \G$ satisfying $\Phi(\gamma)=\pi(\delta)=\pi(\delta')$, by \eqref{eq:bracoid}, we have
  \begin{eqnarray*}
   \gamma\rhd(\delta\cdot \delta')&=& \inv(\gamma)\cdot (\gamma\star  (\delta\cdot \delta'))\\
   &=&\inv(\gamma)\cdot\Big((\gamma\star \delta)\cdot \inv(\gamma)\cdot (\gamma\star \delta')\Big)\\
   &=&\Big(\inv(\gamma)\cdot(\gamma\star \delta)\Big)\cdot \Big(\inv(\gamma)\cdot (\gamma\star \delta')\Big)\\
   &=&  ( \gamma\rhd\delta) \cdot (\gamma\rhd\delta'),
  \end{eqnarray*}
  which implies that Axiom (ii)   in Definition \ref{defi:post-groupoid} holds.

  Finally, for all $(\gamma_1,\gamma_2,\gamma_3)\in \G_\Phi\times_\pi \G_\Phi\times_\pi \G$, we have
  \begin{eqnarray*}
\gamma_1\rhd(\gamma_2\rhd\gamma_3)&=&\inv(\gamma_1)\cdot(\gamma_1\star(\gamma_2\rhd\gamma_3))\\
&=&\inv(\gamma_1)\cdot(\gamma_1\star(\inv(\gamma_2)\cdot(\gamma_2\star\gamma_3)))\\
&=&\inv(\gamma_1)\cdot\Big((\gamma_1\star\inv(\gamma_2))\cdot\inv(\gamma_1)\cdot(\gamma_1\star(\gamma_2\star\gamma_3))\Big)\\
&=&( \gamma_1\rhd\inv(\gamma_2))\cdot\inv(\gamma_1)\cdot\Big((\gamma_1\star\gamma_2)\star\gamma_3\Big)\\
&=& \inv(\gamma_1\rhd\gamma_2)\cdot\inv(\gamma_1)\cdot\Big((\gamma_1\star\gamma_2)\star\gamma_3\Big)\\
&=& \inv(\gamma_1\cdot(\gamma_1\rhd\gamma_2)) \cdot\Big((\gamma_1\cdot(\gamma_1\rhd\gamma_2))\star\gamma_3\Big)\\
&=&(\gamma_1\cdot(\gamma_1\rhd\gamma_2))\rhd\gamma_3,
  \end{eqnarray*}
   which implies that Axiom (iii)   in Definition \ref{defi:post-groupoid} holds.

   Therefore, $(\G\stackrel{\pi}{\longrightarrow}M,\Phi,\rhd)$ is a post-groupoid.
\end{proof}

\emptycomment{
\begin{thm}\mlabel{functor} Proposition   gives a functor $\FF:\PG\lon
\SLB$, and Proposition
gives  a functor $\GG:\SLB\lon
  \PG$. Moreover, they give an isomorphism between the categories $\PG$  and $\SLB$.
\end{thm}
}

Due to the equivalence of post-groupoids and skew-left bracoids, it is obvious that skew-left bracoids also give rise to quiver-theoretical solutions of the Yang-Baxter equation.

\begin{cor}
 Let a group bundle $(\G\stackrel{\pi}{\longrightarrow}M,\cdot,\iota,\inv)$ and a groupoid $(\xymatrix{ \G \ar@<0.5ex>[r]^{\alpha=\pi} \ar[r]_{\beta} & M},\star,\iota,\overline{\inv})$ be a skew-left bracoid. Then $R_\G:\G_\beta\times_\pi \G\lon \G_\beta\times_\pi \G$ defined by
\begin{eqnarray}
  R_\G(\gamma,\delta)=\big(\inv(\gamma)\cdot(\gamma\star\delta),\overline{\inv}\big(\inv(\gamma)\cdot(\gamma\star\delta)\big)\star \gamma\star \delta\big), \quad \forall (\gamma,\delta)\in \G_\beta\times_\pi \G,
\end{eqnarray}
is a quiver-theoretical solution of the Yang-Baxter equation on the quiver $\xymatrix{ \G \ar@<0.5ex>[r]^{\alpha=\pi} \ar[r]_{\beta} & M}$.
\end{cor}

\begin{proof}
  It follows from Theorem \ref{thm:bracoid-post} and Theorem \ref{pgybe} directly.
\end{proof}

\section{Post-Lie groupoids and post-Lie algebroids} \label{sec:6}
In this section, first we recall post-Lie algebroids, and establish the relation between post-Lie algebroids and relative Rota-Baxter operators on Lie algebroids. Then we show that a post-Lie groupoid naturally gives rise to a post-Lie algebroid via differentiation.
\subsection{Post-Lie algebroids and relative Rota-Baxter operators on Lie algebroids}

\begin{defi} \mcite{Mkz:GTGA}
  A {\bf Lie algebroid} structure on a vector bundle $A\longrightarrow M$ is
a pair consisting of a Lie algebra structure $[\cdot,\cdot]_A$ on
the section space $\Gamma(A)$ and a vector bundle morphism
$a_A:A\longrightarrow TM$
from $A$ to the tangent bundle $TM$, called the {\bf anchor}, satisfying the relation
\begin{equation}\mlabel{eq:LAf}~[x,fy]_A=f[x,y]_A+a_A(x)(f)y,\quad \forall x,y\in\Gamma(A), f\in
\CWM.\end{equation}
\end{defi}
We  denote a Lie algebroid by $(A,[\cdot,\cdot]_A,a_A)$ or simply $A$
if there is no danger of confusion.

\begin{ex}
 Let $\phi: \mathfrak{g}\longrightarrow \mathfrak{X}(M)$ be a left action of a Lie algebra  $\mathfrak{g}$   on a manifold $M$, that is, a Lie algebra homomorphism from $(\g,[\cdot,\cdot]_\g)$ to the Lie algebra $(\frkX(M),[\cdot,\cdot]_{TM})$ of vector fields. Consider the trivial vector bundle $A=M\times\g$. Then the section space $\Gamma(M\times\g)$ is isomorphic to $C^\infty(M)\otimes \g$. An element $f\otimes u$ will be simply wrote as $fu$. We have a Lie algebroid structure on  $A$,
whose anchor $a_A:M\times\g\lon TM$ and  Lie bracket $[\cdot,\cdot]_A:\wedge^2(C^\infty(M)\otimes \g)\lon C^\infty(M)\otimes \g$ are given by
\begin{eqnarray}\mlabel{action1}
a_A(m,u)&=&\phi(u)_m,\quad \forall m\in M, u\in\g,\\
  \mlabel{action2}{[fu,gv]}_A&=&fg[u,v]_{\mathfrak{g}}+f\phi(u)(g)v-g\phi(v)(f)u,\quad \forall u,v\in \mathfrak{g},f,g\in C^\infty(M).
\end{eqnarray}
This Lie algebroid is called the {\bf action Lie algebroid} \mcite{Mkz:GTGA}.
\end{ex}

\begin{defi} \mcite{Val}
A {\bf post-Lie algebra} $(\g,[\cdot,\cdot]_\g,\triangleright)$ consists of a Lie algebra $(\g,[\cdot,\cdot]_\g)$ and a binary product $\triangleright:\g\otimes\g\lon\g$ such that
\begin{eqnarray}
\mlabel{Post-1}u\triangleright[v,w]_\g&=&[u\triangleright v,w]_\g+[v,u\triangleright w]_\g,\\
 \mlabel{Post-22}[u,v]_\g\triangleright w&=&a_\triangleright(u,v,w)-a_\triangleright(v,u,w),
\end{eqnarray}
here $a_{\triangleright}(u,v,w):=u\triangleright(v\triangleright w)-(u\triangleright v)\triangleright w $ and $u,v,w\in \g.$
\end{defi}

Define $L_\triangleright:\g\lon \gl(\g)$ by $L_\triangleright(u)(v)=u\triangleright v$. Then by \eqref{Post-1},    $L_\triangleright$ is a linear map from $\g$ to $\Der(\g)$.

As the geometrization of post-Lie algebras, the notion of  post-Lie algebroids was given in \cite[Definition 2.22]{ML} in the study of geometric numerical analysis.   See \cite{MSV} for more applications.

\begin{defi}\mlabel{defi:postLA}
A {\bf post-Lie algebroid} structure on a vector bundle
$A\longrightarrow M$ is a triple that consists of a $C^\infty(M)$-linear Lie
algebra structure $[\cdot,\cdot]_A$ on $\Gamma(A)$, a  bilinear operation   $\triangleright_A:\Gamma(A)\otimes \Gamma(A)\longrightarrow \Gamma(A)$   and a
vector bundle morphism $a_A:A\longrightarrow TM$, called the {\bf anchor},
such that $(\Gamma(A),[\cdot,\cdot]_A,\triangleright_A)$ is a post-Lie algebra, and
for all $f\in\CWM$ and $u,v\in\Gamma(A)$, the following
relations are satisfied:
\begin{itemize}
\item[\rm(i)]$u\triangleright_A(fv)=f(u\triangleright_A v)+a_A(u)(f)v,$
\item[\rm(ii)] $(fu)\triangleright_A v=f(u\triangleright_A v).$
\end{itemize}
\end{defi}

We usually denote a post-Lie algebroid by $(A,[\cdot,\cdot]_A,\triangleright_A, a_A)$.  If the Lie algebra structure $[\cdot,\cdot]_A$ in a post-Lie algebroid $(A, [\cdot,\cdot]_A, \triangleright_A, a_A)$ is abelian, then it becomes a left-symmetric algebroid, which is also called a Koszul-Vinberg algebroid. See \mcite{Boyom2,LSBC,ML} for more details.

 The following result is well known, e.g. see \cite[Proposition 2.23]{ML}.
  \begin{pro}
  Let  $(A,[\cdot,\cdot]_A,\triangleright_A, a_A)$ be a post-Lie algebroid. Then $(A,[\cdot,\cdot]_{\triangleright_A},  a_A)$ is a Lie algebroid, called the {\bf Grossman-Larson Lie algebroid}, where the   Lie bracket $[\cdot,\cdot]_{\triangleright_A}$ is given by
  \begin{equation}\label{eq:subadj-algebroid}
    [x,y]_{\triangleright_A}=[x,y]_A+x\triangleright_A y-y\triangleright_A x,\quad\forall x,y\in\Gamma(A).
  \end{equation}
  \end{pro}

\begin{ex}\label{ex:action-post-algebroid}

Let $(\g,[\cdot,\cdot]_\g)$ be a Lie algebra and $M$ a manifold. Then the section space $\Gamma(M\times \g)$ of the trivial bundle $A=M\times \g$  enjoys a  $C^\infty(M)$-linear Lie algebra structure $[\cdot,\cdot]_\g$ given by
$$
[fu,gv]_\g=fg[u,v]_\g,\quad \forall f,g\in C^\infty(M), ~u,v\in\g.
$$
Let $\phi:\g\to\frkX(M)$ be a left action of $\g$ on $M$. Then one can define $a_A:M\times\g\to TM$ by \eqref{action1}, and define $\triangleright_A:\Gamma(A)\times \Gamma(A)\longrightarrow \Gamma(A)$ by
\begin{equation}\label{eq:act-post-Lie-product}
  (fu)\triangleright_A (gv)=f\phi(u)(g)v,\quad \forall f,g\in C^\infty(M), ~u,v\in\g.
\end{equation}
Then $(A=M\times \g,[\cdot,\cdot]_\g,\triangleright_A,a_A)$ is a post-Lie algebroid.  See \cite[Proposition 4.3]{MSV} for more details.
 \end{ex}

 \begin{rmk}
 We call $(A=M\times \g,[\cdot,\cdot]_\g,\triangleright_A,a_A)$  the {\bf MKW post-Lie algebroid} of a Lie algebra action on a smooth manifold. The corresponding  post-Lie algebraic structure on the section space $\Gamma(A)=C^{\infty}(M,\g)$ is called the {\bf MKW post-Lie algebra}.
\end{rmk}

 \begin{defi}
    An {\bf action of a post-Lie algebra} $(\g,[\cdot,\cdot]_\g,\triangleright_\g)$ on a manifold $M$ is a Lie algebra homomorphism $\phi:(\g,[\cdot,\cdot]_{\triangleright_\g})\lon \frkX(M).$
 \end{defi}

 Let $\phi:(\g,[\cdot,\cdot]_{\triangleright_\g})\lon \frkX(M)$ be an {action of the post-Lie algebra} $(\g,[\cdot,\cdot]_\g,\triangleright_\g)$ on a manifold $M$. On the trivial bundle $A=M\times \g$, define the anchor $a_A:M\times\g\lon TM$ by \eqref{action1} and define the bilinear operation $\triangleright_A:\otimes^2(C^\infty(M)\otimes \g)\lon C^\infty(M)\otimes \g$   by
\begin{eqnarray}
 \mlabel{action21}  (fu)\triangleright_A gv  &:=&fgu\triangleright_\g v+f\phi(u)(g)v,\quad \forall u,v\in \mathfrak{g},f,g\in C^\infty(M).
\end{eqnarray}

\begin{pro}\cite{GLS}\label{pro:actionPL}
  With the above notations, $(M\times\g, [\cdot,\cdot]_\g,\triangleright_A,a_A)$ is a post-Lie algebroid, called the {\bf action post-Lie algebroid} of the post-Lie algebra $(\g,[\cdot,\cdot]_\g,\rhd)$.
\end{pro}

 Now we introduce the notion of a relative Rota-Baxter operator on a Lie algebroid with respect to an action on a bundle of Lie algebras. Recall that a bundle of Lie algebras is a vector bundle such that on each fibre there is a Lie algebra structure which are not necessarily isomorphic. While for a Lie algebra bundle, the Lie algebra structure on each fibre are isomorphic \cite[Definition 3.3.8]{Mkz:GTGA}. In the definition of a post-Lie algebroid, the requirement that there is a $C^\infty(M)$-linear Lie
algebra structure $[\cdot,\cdot]_A$ on $\Gamma(A)$ is equivalent to the fact that $(A,[\cdot,\cdot]_A)$ is a bundle of Lie algebras.

 \begin{defi}\label{defi:action-algebroid}
   An {\bf action of a Lie algebroid} $(A\stackrel{}{\longrightarrow}M,[\cdot,\cdot]_A,a_A)$ on a bundle of Lie algebras $(E\stackrel{ }{\longrightarrow}M,[\cdot,\cdot]_E)$ is a bilinear map $\psi:\Gamma(A)\times \Gamma(E)\to \Gamma(E)$ such that
   \begin{itemize}
     \item[(i)]$ \psi(fx,u)=f\psi(x,u),$
     \item[(ii)]$\psi(x,fu)=f\psi(x,u)+a_A(x)(f)u,$
      \item[(iii)]$\psi(x,[u,v]_E)=[\psi(x,u),v]_E+[u,\psi(x,v)]_E,$
       \item[(iv)]$\psi([x,y]_A,u)=\psi(x,\psi(y,u))-\psi(y,\psi(x,u)),$
   \end{itemize}
   for all $x,y\in\Gamma(A), ~u,v\in\Gamma(E)$ and $f\in C^\infty(M)$.
 \end{defi}

  Let
$E{\longrightarrow}M$ be a vector bundle and $\dev E$ the associated covariant
differential operator bundle. Any section $\frkd\in\Gamma(\dev E)$ satisfies
$$
\frkd(fu)=f\frkd(u)+\frka(\frkd)(f)u,\quad\forall u\in\Gamma(E),~f\in C^\infty(M),
$$
for a bundle map $\frka:\dev E\to TM$, which is called the anchor.
It is well known that $(\dev E,[\cdot,\cdot]_\dev,\frka)$ is a transitive Lie algebroid, where $[\cdot,\cdot]_\dev$ is the commutator bracket (see \cite[Example
3.3.4]{Mkz:GTGA}).   Let $(E{\longrightarrow}M,[\cdot,\cdot]_E)$ be a  Lie algebra bundle. Denote by $\Der(E) $ the subalgebroid of $\dev E$, whose section is also a derivation on the Lie algebra $(\Gamma(E),[\cdot,\cdot]_E)$.

The following result is obvious.

\begin{pro}
 An   action of a Lie algebroid  $(A\stackrel{}{\longrightarrow}M,[\cdot,\cdot]_A,a_A)$ on a Lie algebra bundle $(E{\longrightarrow}M,[\cdot,\cdot]_E)$ is equivalent to a homomorphism from the Lie algebroid $A$ to $\Der(E)$.
\end{pro}

 \begin{defi}
  Let $\psi:\Gamma(A)\times \Gamma(E)\to \Gamma(E)$ be an action of a Lie algebroid  $(A\stackrel{}{\longrightarrow}M,[\cdot,\cdot]_A,a_A)$ on a bundle of Lie algebras $E\stackrel{ }{\longrightarrow}M$. A base-preserving bundle map $B:E\to A$ is called a {\bf relative Rota-Baxter operator} if
  $$
  [B(u),B(v)]_A=B(\psi(B(u),v)-\psi(B(v),u)+[u,v]_E),\quad\forall u,v\in \Gamma(E).
  $$
 \end{defi}

  \begin{thm}\label{thm:postalg-RB}
  Let  $(A,[\cdot,\cdot]_A,\triangleright_A, a_A)$ be a post-Lie algebroid. Then $\triangleright_A$ gives rise to an action of the Grossman-Larson Lie algebroid $(A,[\cdot,\cdot]_{\triangleright_A},  a_A)$ on the bundle of Lie algebras $(A,[\cdot,\cdot]_A)$, and the identity map $\Id:A\to A$ is a relative Rota-Baxter operator on the Grossman-Larson Lie algebroid $(A,[\cdot,\cdot]_{\triangleright_A},  a_A)$ with respect to the action $\triangleright_A$   on the bundle of Lie algebras  $(A,[\cdot,\cdot]_A)$.
  \end{thm}

  \begin{proof}

    Then Axioms (i)-(iv) in Definition \ref{defi:action-algebroid} for $\triangleright_A$ follows from  Axioms (i)-(ii) in Definition \ref{defi:postLA}, \eqref{Post-1} and \eqref{Post-2} respectively. Thus $\triangleright_A$ gives rise to an action.

    Finally, by \eqref{eq:subadj-algebroid}, the identity map $\Id:A\to A$ is a relative Rota-Baxter operator on the Grossman-Larson Lie algebroid $(A,[\cdot,\cdot]_{\triangleright_A},  a_A)$ with respect to the action $\triangleright_A$.
  \end{proof}

  Let $B:E\to A$ be a  relative Rota-Baxter operator on a Lie algebroid     $(A\stackrel{}{\longrightarrow}M,[\cdot,\cdot]_A,a_A)$ with respect to an action $\psi:\Gamma(A)\times \Gamma(E)\to \Gamma(E)$ on a bundle of Lie algebras $E\stackrel{ }{\longrightarrow}M$. Define $\triangleright_E:\Gamma(E)\otimes \Gamma(E)\to \Gamma(E)$ by
  \begin{eqnarray}
   u \triangleright_Ev=\psi(B(u),v),\quad\forall u,v\in\Gamma(E).
  \end{eqnarray}

  \begin{thm}\label{thm:RB-postalg}
    With the above notations, $(E,[\cdot,\cdot]_E,\triangleright_E,a_A\circ B)$ is a post-Lie algebroid.
  \end{thm}
  \begin{proof}
   For all $u,v\in\Gamma(E),~f\in C^\infty(M)$, by Axiom (i) in Definition \ref{defi:action-algebroid} and the fact that $B$ is a bundle map, we have
    $$
    (fu) \triangleright_Ev=\psi(B(fu),v)=\psi(fB(u),v)=f\psi(B(u),v)=f(u \triangleright_Ev).
    $$

    By Axiom (ii) in Definition \ref{defi:action-algebroid}, we have
    $$
    u \triangleright_Efv=\psi(B(u),fv)=f\psi(B(u),v)+a_A(B(u))(f)v=f(u\triangleright_Ev)+a_A(B(u))(f)v.
    $$

    By Axiom (iii) in Definition \ref{defi:action-algebroid}, we have
    \begin{eqnarray*}
   u \triangleright_E[v,w]_E=\psi(B(u),[v,w]_E)=[\psi(B(u),v),w]_E+[v,\psi(B(u),w)]_E=[u\triangleright_Ev,w]_E+[v, u\triangleright_Ew]_E.
    \end{eqnarray*}

    Finally,  by Axiom (iv) in Definition \ref{defi:action-algebroid}, we have
  \begin{eqnarray*}
   a_{\triangleright_E}(u,v,w)-a_{\triangleright_E}(v,u,w)&=&u\triangleright_{E}(v\triangleright_{E} w)-(u\triangleright_{E} v)\triangleright_{E} w  \\
   &&-v\triangleright_{E}(u\triangleright_{E} w)+(v\triangleright_{E} u)\triangleright_{E} w\\
   &=&\psi(B(u),\psi(B(v),w))-\psi(B(v),\psi(B(u),w))\\
   &&-\psi(B(\psi(B(u),v)),w)+\psi(B(\psi(B(v),u)),w)\\
   &=&\psi([B(u),B(v)]_A-B(\psi(B(u),v))+B(\psi(B(v),u)),w)\\
   &=&\psi(B[u,v]_E,w)=[u,v]_E\triangleright_Ew.
   \end{eqnarray*}

   Therefore, $(E,[\cdot,\cdot]_E,\triangleright_E,a_A\circ B)$ is a post-Lie algebroid.
  \end{proof}

  \begin{cor}
    Let $B:E\to A$ be a  relative Rota-Baxter operator on a Lie algebroid     $(A,[\cdot,\cdot]_A,a_A)$ with respect to an action $\psi:\Gamma(A)\otimes \Gamma(E)\to \Gamma(E)$ on a bundle of Lie algebras $E$. Then $(E,[\cdot,\cdot]_B,a_A\circ B)$ is a Lie algebroid, called the {\bf descendent Lie algebroid}, where the Lie bracket $[\cdot,\cdot]_B$ is given by
  \begin{equation}
    [u,v]_B=[u,v]_E+\psi(B(u),v)-\psi(B(v),u),\quad\forall u,v\in\Gamma(E).
  \end{equation}

  Moreover, $B:E\to A$ is a homomorphism from the descendent Lie algebroid $(E,[\cdot,\cdot]_B,a_A\circ B)$ to the Lie algebroid     $(A,[\cdot,\cdot]_A,a_A)$.
  \end{cor}
  \begin{proof}
    It follows from Theorem \ref{thm:RB-postalg} and Theorem \ref{thm:postalg-RB} directly.
  \end{proof}

\subsection{Differentiation of post-Lie groupoids}

 The tangent space of a Lie group at the identity has a Lie algebra structure. As its geometrization, on the vector bundle $A:=\ker(\alpha_*)|_M\longrightarrow M$ from a Lie groupoid $\xymatrix{ (\G \ar@<0.5ex>[r]^{\alpha} \ar[r]_{\beta} & M},\cdot,\iota,\inv)$,    there is a Lie algebroid structure defined as follows \cite{Mkz:GTGA}:
the anchor map $a_A:A\longrightarrow TM$ is simply $\beta_*$ and  the Lie bracket $[u,v]_A$ is determined by
\[\overleftarrow{[u,v]_A}=[\overleftarrow{u},\overleftarrow{v}]_{T\G},\qquad \forall u,v\in \Gamma(A),\]
where $\overleftarrow{u}$ denotes the left-invariant vector field on $\G$ given by $\overleftarrow{u}_\gamma=L_{\gamma_*}u_{\beta(\gamma)}$.

The differentiation of a Lie group bundle  $\G\stackrel{\pi}{\longrightarrow}M$ is a Lie algebra bundle $A\longrightarrow M$, where the fiber $A_m$ is the Lie algebra of the Lie group $\G_m$, see \cite[Example 3.5.12]{Mkz:GTGA}. Let $\exp$ be the usual exponential map from Lie algebras to Lie groups. Then there is a well defined map, also denoted by $\exp$ from $\Gamma(A)$ to $\Gamma(\G)$ defined by
 $
 \exp(u)(m)=\exp(u_m),$ for all $u\in \Gamma(A), m\in M.
 $
 Obviously there holds:
 $
 u=\frac{d}{dt}\eval{t=0}\exp(tu).
 $

 Denote by $\Aut(A)$ the Lie groupoid whose space of objects is $M$, and the space of morphisms from $m$ to $n$ is Lie algebra isomorphisms from $A_n$ to $A_m$. The Lie algebroid of the Lie groupoid $\Aut(A)$ is exactly $\Der(A)$.

Let $(\G\stackrel{\pi}{\longrightarrow}M,\cdot,\Phi,\rhd)$ be a post-Lie groupoid. Then $L^\rhd_\gamma:\G_{\Phi(\gamma)}\to \G_{\pi(\gamma)}$ is an isomorphism of Lie groups. Taking the differentiation, for which we use the same notation, we obtain an isomorphism $L^\rhd_\gamma:A_{\Phi(\gamma)}\to A_{\pi(\gamma)}$. Consequently, by Theorem \ref{thm:subadj}, we obtain a Lie groupoid homomorphism  $L^\rhd$ from the  Lie groupoid $(\xymatrix{ \G \ar@<0.5ex>[r]^{\alpha=\pi} \ar[r]_{\beta=\Phi} & M},\star,\iota,\inv_\rhd)$ to $\Aut(A)$. 
 Denote the Lie algebroid of the Grossman-Larson Lie groupoid $(\xymatrix{ \G \ar@<0.5ex>[r]^{\alpha=\pi} \ar[r]_{\beta=\Phi} & M},\star,\iota,\inv_\rhd)$   by $(A\longrightarrow M,\Courant{\cdot,\cdot},\Phi_*)$. The Lie groupoid homomorphism $L^\rhd$ induces a Lie algebroid homomorphism $L^\rhd_*$ from $(A\longrightarrow M,\Courant{\cdot,\cdot},\Phi_*)$ to $\Der(A)$. In particular, there holds:
\begin{equation}\label{eq:anc}
  \Phi_*(u)=\frka(L^\rhd_*(u)),\quad\forall u\in \Gamma(A).
\end{equation}

 Define $\triangleright_A:\Gamma(A)\otimes \Gamma(A) \to \Gamma(A) $ by
\begin{equation}\label{eq:diffrhd}
u\triangleright_A v= L^\rhd_*(u)v,\quad\forall u,v\in\Gamma(A).
\end{equation}

\begin{thm}\label{thm:diff-post}
 Let $(\G\stackrel{\pi}{\longrightarrow}M,\cdot,\Phi,\rhd)$ be a post-Lie groupoid. Then $(A\longrightarrow M,[\cdot,\cdot]_A,\Phi_*,\triangleright_A)$ is a post-Lie algebroid, where $(A\longrightarrow M,[\cdot,\cdot]_A)$ is the Lie algebra bundle associated to the Lie group bundle $\G\stackrel{\pi}{\longrightarrow}M$, and $\triangleright_A$ is given by \eqref{eq:diffrhd}.
\end{thm}

\begin{proof}
Using the weak post-Lie group structure associated to a post-Lie groupoid given in Theorem \ref{thm:bisection}, we have
$$
u\triangleright_A v= \frac{d}{dt}\eval{t=0} \frac{d}{ds}\eval{s=0} L^\rhd_{\exp(tu)}\exp(sv),\quad\forall u,v\in\Gamma(A).
$$
More precisely,
$$
(u\triangleright_A v)_m= \frac{d}{dt}\eval{t=0} L^\rhd_{\exp(tu_m)}v_{\Phi(\exp(tu_m))}=\frac{d}{dt}\eval{t=0}  \frac{d}{ds}\eval{s=0} L^\rhd_{\exp(tu_m)}\exp(sv_{\Phi(\exp(tu_m))}),\quad\forall u,v\in\Gamma(A).
$$
Then similar as the computation in \cite[Theorem 4.3]{PostG}, we have
$$
u\triangleright_A (v\triangleright_A w)- v\triangleright_A (u\triangleright_A w)=[u,v]_{\triangleright_A}\triangleright_A w,
$$
which implies that $(\Gamma(A),[\cdot,\cdot]_A,\triangleright_A)$ is a post-Lie algebra.

For all $u,v\in\Gamma(A)$ and $f\in C^\infty(M)$, by \eqref{eq:anc}, we have
\begin{eqnarray*}
  (fu)\triangleright_A v&=&L^\rhd_*(fu)v=fL^\rhd_*(u)v=f(u\triangleright_Av),\\
  u\triangleright_A(fv)&=&L^\rhd_*(u)(fv)=fL^\rhd_*(u)v+\frka(L^\rhd_*(u))(f)v=f(u\triangleright_A v)+\Phi_*(u)(f)v,
\end{eqnarray*}
which implies that  $(A\longrightarrow M,[\cdot,\cdot]_A,\Phi_*,\triangleright_A)$ is a post-Lie algebroid.
\end{proof}

Let $\Phi:M\times G\to M$ be a right action of a Lie group $G$ on $M$. Then it induces a left action $\phi:\g\to\frkX(M)$ of the corresponding Lie algebra $\g=T_eG$ on $M$ via\footnote{In general, a right action of a Lie group induces a right action of the corresponding Lie algebra. Here we consider the Lie algebra to be obtained by right invariant vector fields, so that a right action of a Lie group induce a left action of the corresponding Lie algebra.}
$$
\phi(u)_m=\Phi_{*m}(u)=\frac{d}{dt}\eval{t=0}\Phi(m,\exp(tu)),\quad\forall u\in\g,~m\in M.
$$
\begin{pro}
  The differentiation of the MKW post-Lie groupoid obtained from a Lie group action given in Proposition \ref{ex:post-g} is the MKW post-Lie algebroid given in Example \ref{ex:action-post-algebroid}.
  \end{pro}
  \begin{proof}
    First it is obvious that the differentiation of the trivial Lie group bundle $M\times G\to M$ is the trivial Lie algebra bundle $A=M\times \g\to M.$

    It is also straightforward to see that the anchor map $a_A(u)_m=\phi(u)_m=\Phi_{*m}(u)$.

    For all $u,v\in\g$, which are viewed as constant sections of $\Gamma(A)$, we have
    \begin{eqnarray*}
  (u\triangleright_A v)_m  &=&\frac{d}{dt}\eval{t=0}  \frac{d}{ds}\eval{s=0} L^\rhd_{\exp(tu_m)}\exp(sv_{\Phi(\exp(tu_m))})\\
  &=& \frac{d}{dt}\eval{t=0} \frac{d}{ds}\eval{s=0}(m,\exp(tu_m))\rhd (\Phi(\exp(tu_m)),\exp(sv_{\Phi(\exp(tu_m))}))\\
   &=& \frac{d}{dt}\eval{t=0} \frac{d}{ds}\eval{s=0}(m,\exp(sv_{\Phi(\exp(tu_m))}))\\
   &=&\frac{d}{dt}\eval{t=0}  (m, v_{\Phi(\exp(tu_m))})=0.
       \end{eqnarray*}
       Therefore, by Axioms (i) and (ii) in Definition \ref{defi:postLA}, the induced post-Lie product must be given by \eqref{eq:act-post-Lie-product}, which implies that the differentiation of the post-Lie groupoid given in Example \ref{ex:post-g} is the post-Lie algebroid given in Example \ref{ex:action-post-algebroid}.
  \end{proof}

  \begin{pro}
    The differentiation of the post-Lie groupoid given in Proposition \ref{pro:act-post-groupoid} is the post-Lie algebroid given in  Proposition \ref{pro:actionPL}.
  \end{pro}

  \begin{proof}
  Similar as the proof of the above proposition,  for all $u,v\in\g$, which are viewed as constant sections of $\Gamma(A)$, we have
    \begin{eqnarray*}
  (u\triangleright_A v)_m  &=&\frac{d}{dt}\eval{t=0}  \frac{d}{ds}\eval{s=0} L^{\bar{\rhd}}_{\exp(tu_m)}\exp(sv_{\Phi(\exp(tu_m))})\\
  &=& \frac{d}{dt}\eval{t=0} \frac{d}{ds}\eval{s=0}(m,\exp(tu_m))\bar{\rhd} (\Phi(\exp(tu_m)),\exp(sv_{\Phi(\exp(tu_m))}))\\
   &=& \frac{d}{dt}\eval{t=0} \frac{d}{ds}\eval{s=0}(m,\exp(tu)\rhd \exp(sv))=(m, u\triangleright_\g v).
       \end{eqnarray*}
        Therefore, by Axioms (i) and (ii) in Definition \ref{defi:postLA}, the induced post-Lie product must be given by \eqref{action21}, which implies that the differentiation of the post-Lie groupoid given in Proposition \ref{pro:act-post-groupoid} is the post-Lie algebroid given in  Proposition \ref{pro:actionPL}.
  \end{proof}

\noindent
{\bf Acknowledgements. } This research is supported by NSFC (12371029,12471060).

\end{document}